\documentclass[a4paper]{jpconf}
\usepackage{graphicx}

\usepackage{amsmath,amsfonts,amssymb,dsfont}
\usepackage{amsthm}
\usepackage{psfrag}
\usepackage{wrapfig}
\usepackage{sidecap}
\usepackage{cite}

\newtheorem{theorem}{Theorem}

\newtheorem{lemma}{Lemma}
\newtheorem{pr}{Rule}

\newtheorem*{claim}{Claim}

\def\be{\begin{equation}}
\def\ee{\end{equation}}
\def\ba{\begin{eqnarray}}
\def\ea{\end{eqnarray}}

\newcommand\nn{\nonumber}
\newcommand\q{\quad}

\newcommand{\cq}{\mathcal Q}

\newcommand{\ct}{\mathcal T}

\newcommand{\fo}{\mathfrak{o}}

\newcommand{\fs}{\mathfrak{s}}

\def\nn{\nonumber}

\newcommand{\eqa}{\begin{eqnarray}}
\newcommand{\neqa}{\end{eqnarray}}

\newcommand{\p}{\partial}

\def\f{\frac}

\usepackage{bbm}

\newcommand{\SU}{\mathrm{SU}}

\def\q{{\quad}}

\begin{document}

\vspace*{-3cm}
\title{Quantum theory from rules on information acquisition}

\author{Philipp Andres H\"ohn}

\address{Vienna Center for Quantum Science and Technology, and Institute for Quantum Optics and Quantum Information, Austrian Academy of Sciences, Boltzmanngasse 3, 1090 Vienna, Austria}

\ead{p.hoehn@univie.ac.at}

\begin{abstract}
We summarise a recent reconstruction of the quantum theory of qubits from rules constraining an observer's acquisition of information about physical systems. This review is accessible and fairly self-contained, focussing on the main ideas and results and not the technical details. The reconstruction offers an informational explanation for the architecture of the theory and specifically for its correlation structure. In particular, it explains entanglement, monogamy and non-locality compellingly from limited accessible information and complementarity. As a by-product, it also unravels new `conserved informational charges' from complementarity relations that characterise the unitary group and the set of pure states. 
\end{abstract}

\section{Introduction}

Why is the physical world described by quantum theory? 
If we wish to sensibly address this question we have to step beyond quantum theory and to consider it within a landscape of alternative theories. This, after all, permits us to ponder about how the world {\it could} have been different, possibly described by modifications of quantum theory.
Such an endeavour forces us to leave the usual textbook formulation of quantum theory -- and everything we take for granted about it -- behind and to develop a more general language that also applies to alternative theories. Ideally, this language should be operational, encompassing the interactions of some observer with physical systems in a plethora of conceivable, physically distinct worlds. 

If we wish to also provide a possible answer to the above question, we then have to find {\it physical} properties of quantum theory that single it out, at least within the given landscape of alternatives. In particular, the goal should be to find an operational justification for the textbook axioms, i.e., ultimately for complex Hilbert spaces, unitary dynamics, tensor product structure for composite systems, Born rule, and so on. The result would be a reconstruction of quantum theory from operational axioms \cite{Hoehn:2014uua,hw,Hardy:2001jk,Dakic:2009bh,masanes2011derivation,chiribella2011informational,Barnum,M2,2008arXiv0805.2770G,Fuchs} and should ideally yield a better understanding of what quantum theory tells us about Nature -- and why it is the way it is.

In this manuscript, we shall review and summarise how the quantum formalism for arbitrarily many qubits can be reconstructed from operational rules restricting an observer's acquisition of information about a set of observed systems \cite{Hoehn:2014uua,hw}. The goal of this summary is to provide a didactical and easily accessible overview over this reconstruction. Its underlying framework is especially engineered for unraveling the architecture of quantum theory and so many reconstruction steps are instructive for understanding the origin of quantum properties. As we shall see, this reconstruction provides a transparent, informational explanation for the structure of qubit quantum theory and especially also for its paradigmatic features such as entanglement, monogamy and non-locality. The approach also produces novel `conserved informational charges', indeed appearing in quantum theory, that turn out to characterise the unitary group and the set of pure states and which might find practical applications in quantum information. %In conjunction, this gives rise to an informational

The premise of the summarised approach is to only speak about information that the observer has access to. It is thus purely operational and survives without any ontological commitments. This approach is inspired, in part, by Rovelli's {\it relational quantum mechanics} \cite{Rovelli:1995fv} and the Brukner-Zeilinger informational interpretation of quantum theory \cite{zeilinger1999foundational,Brukner:2002kx}; this successful reconstruction can be viewed as a completion of these ideas for qubit systems.

The rest of the manuscript is organized as follows. In sec.~\ref{sec_post}, we review the landscape of alternative theories, in sec.~\ref{sec_axioms}, we formulate the operational quantum axioms, in sec.~\ref{sec_reconst}, we summarise the key steps of the reconstruction itself and, finally, conclude in sec.~\ref{sec_conc}.

\section{Overview over a landscape of theories}\label{sec_post}

We shall begin with an overview over a landscape of alternative theories which has been developed in \cite{Hoehn:2014uua,hw} to which we also refer for further details.

\subsection{From questions and answers to probabilities and states}\label{sec_states}
Our first aim is to define a notion of state both for a single system and an ensemble of systems.

Consider an observer $O$ who interrogates an ensemble of (identically prepared \cite{Hoehn:2014uua}) systems $\{S_a\}_{a=1}^n$, coming out of a preparation device, with binary questions $Q_i$ from some set $\cq$. For example, in the case of quantum theory, such a question could read ``is the spin of the electron up in $x$-direction?." This set $\cq$ shall only contain {\it repeatable} questions in the sense that $O$ will receive $m\in\mathbb{N}$ times the same answer whenever asking any $Q_i\in\cq$ $m$ times in immediate succession to a single system $S_a$. We shall assume any $S_a$ to always give a definite answer if asked some $Q_i\in \cq$ which moreover is not independent of $S_a$'s preparation. Accordingly, $\cq$ can only contain physically implementable questions which are `answerable' by the $\{S_a\}$ and not arbitrary logically conceivable binary questions. Also, since we assume definite answers we do not address the measurement problem. The answers to the $Q_i\in\cq$ given by the $\{S_a\}$ shall follow a specific statistics for each way of preparing the $\{S_a\}$ (for $n$ sufficiently large). The set of all the possible answer statistics for all $Q_i\in\cq$ for all preparations is denoted by $\Sigma$.

$O$, being a good experimenter, has developed, through his experiments, a theoretical model for $\cq$ and $\Sigma$ which he employs to interpret the outcomes of his interrogations (and to decide whether a question is in $\cq$ or not). This permits $O$ to assign, for the next $S_a$ to be interrogated, a prior probability $y_i$ that $S_a$'s answer to $Q_i\in\cq$ will be `yes'. Namely, $O$ determines $y_i$ through a belief updating (in a broadly Bayesian spirit\footnote{We add `broadly' here as we also consider the typical laboratory situation of an ensemble of systems.}) according to his model of $\Sigma$, any prior information on the way of preparation, and possibly to the frequencies of `yes' answers to questions from $\cq$ which he may have recorded in previous interrogation runs on systems identically prepared to $S_a$. In particular, $O$ may also not have carried out previous interrogations on systems identically prepared to $S_a$ (e.g., if the ensemble contains only the single $S_a$) in which case he will estimate the prior $y_i$ for the single $S_a$ solely according to his model of $\Sigma$ and any prior information about the preparation. (More on this and update rules will be discussed in secs.~\ref{sec_noinfo} and \ref{sec_indep}.)

While $\cq$ need not necessarily contain {\it all} binary measurements that $O$ could, in principle, perform on the $\{S_a\}$, we shall assume that $\cq$ is `tomographically complete' in the sense that the $\{y_i\}_{\forall\,Q_i\in\cq}$ are sufficient to compute the probabilities for all other physically realizable measurements possibly not contained in the $\cq$ too. Hence, the $y_i$ encode everything $O$ could possibly say about the future outcomes to arbitrary experiments on the $\{S_a\}$ in his laboratory. It will therefore be sufficient to henceforth restrict $O$ to acquire information about the $S_a$ solely through the $Q_i\in\cq$. It is also natural to identify $O$'s `catalogue of knowledge' about the given $S_a$, i.e.\ the collection of $\{y_i\}_{\forall\,Q_i\in\cq}$, with {\it the state of $S_a$ relative to $O$}. This is a state of information and an element of $\Sigma$. Conversely, any element in $\Sigma$ assigns a probability $y_i$ to all $Q_i\in\cq$. Thus, we identify $\Sigma$ with the {\it state space} of $S_a$.

The state $\{y_i\}_{\forall\,Q_i\in\cq}$ is the prior state for the single $S_a$ to be interrogated next but also coincides with the state $O$ assigns to the ensemble $\{S_a\}$ (which may only contain a single member) given that its members are identically prepared \cite{Hoehn:2014uua}.

\subsection{Time evolution of $O$'s `catalogue of knowledge'}

We permit $O$ to subject the $\{S_a\}$ to interactions which cause a state $\{y_i(t_0)\}_{\forall\,Q_i\in\cq}$ at time $t_0$ to evolve in time to another legitimate state. Any permitted time evolution shall be temporally translation invariant, thus defining a one-parameter map $T_{\Delta t}(\{y_i(t_0)\}_{\forall\,Q_i\in\cq})=\{y_i(t_0+\Delta t)\}_{\forall\,Q_i\in\cq}$ from $\Sigma$ to itself which only depends on the time interval $\Delta t$ but not on $t_0$. We denote by $\ct$ the set of all time evolutions to which we allow $O$ to expose the $\{S_a\}$.

Clearly, $\ct$ is a further crucial ingredient of $O$'s world model; his model for describing his interrogations with the $\{S_a\}$ is thus encoded in the triple $(\cq,\Sigma,\ct)$.

\subsection{Convexity and state of no information }\label{sec_noinfo}

It will be our challenge to unravel what $O$'s world model is. This requires us to subject the triple $(\cq,\Sigma,\ct)$ to a number of further operational conditions that are `natural' in the context of information acquisition with a broadly Bayesian spirit. Upon imposing the quantum postulates, this will turn out to restrict $\cq$ and $\ct$ to incorporate only a `natural' subset of all possible quantum measurements and time evolutions, namely projective binary measurements and unitaries, respectively (rather than arbitrary positive operator-valued measures (POVMs) and completely positive maps). But this suffices for our purposes to reconstruct the textbook quantum formalism. 

To account for the possibility of randomness in the method of preparation, we assume $\Sigma$ to be convex.  Consider a collection of identical systems (i.e., with identical $(\cq,\Sigma,\ct)$) that are not necessarily in identical states and for which $O$ uses a cascade of biased coin tosses to decide which system to interrogate. Then $O$ is enabled to assign a single prior state to this collection which is a convex combination of their individual states.

Next, we assume the existence of a special method of preparation which generates even completely random answer statistics over all $Q_i\in\cq$. This preparation is described by a special state in $\Sigma$, namely $y_i=\f{1}{2}$, $\forall\,Q_i\in\cq$, and shall be called the {\it state of no information}. This distinguished state is a constraint on the {\it pair} $(\cq,\Sigma)$\footnote{E.g., in quantum theory $(\{\text{binary POVMs}\},\{\text{density matrices}\})$ does not satisfy this condition because there exist inherently biased POVMs, while $(\{\text{projective binary measurements}\},\{\text{density matrices}\})$ does.} and plays two crucial roles: it defines (1) the prior state of $S_a$ that $O$ will start with in a Bayesian updating when he has no `prior information' about the $\{S_a\}$ (except what his model $(\cq,\Sigma,\ct)$ is); and (2) an unambiguous notion of {\it(in-)dependence} of questions (cf.\ sec.~\ref{sec_indep}) which otherwise would be state dependent.\footnote{E.g., in quantum theory, the questions $Q_{x_1}=$``Is the spin of qubit 1 up in $x$-direction?'' and $Q_{x_2}=$``Is the spin of qubit 2 up in $x$-direction?" are independent relative to the completely mixed state, however, not relative to a state with entanglement in $x$-direction.}

\subsection{State updating and (in)dependence and compatibility of questions}\label{sec_indep}

There are two kinds of state update rules, one for the state of the ensemble $\{S_a\}$ (which coincides with the prior state assigned to the next $S_a$ to be interrogated) and one for the posterior state of a given ensemble member $S_a$. In a {\it single shot interrogation}, $O$ receives a single $S_a$, assigns a prior state to it according to his prior information (cf.\ sec.~\ref{sec_states}), interrogates it with some questions from $\cq$ (without intermediate re-preparation) and, depending on the answers, updates the prior to a posterior state valid for this specific $S_a$ only. This requires a consistent {\it posterior state update rule} which permits $O$ to update the probabilities $y_i$ for all $Q_i\in\cq$ in a manner that respects the structure of $\Sigma$ and the repeatability of questions (i.e., an answer $Q_i=$ `yes' or `no' must have a posterior $y_i=1$ or $0$ as a consequence, respectively). This is also a belief updating, but about the single $S_a$ and is {\it not} the same as in secs.~\ref{sec_states} and \ref{sec_noinfo}. Specifically, the posterior state of $S_a$ may differ significantly from its prior state if $O$ has experienced an information gain on at least some $Q_i\in\cq$. (This will necessarily happen when complementary questions are involved, see below.) This is the `collapse' of the state: it is merely $O$'s update of information about the specific $S_a$ \cite{Hoehn:2014uua}.  

By contrast, in a {\it multiple shot interrogation}, $O$ carries out a single shot interrogation on each member of an entire (identically prepared \cite{Hoehn:2014uua}) ensemble $\{S_a\}$ to do {\it ensemble state tomography} and estimate the state of the ensemble from his prior information about the preparation and the collection of posterior states from the single shot interrogations. With every further interrogated $S_a$, $O$ updates the ensemble state -- which coincides with the prior state of the next system from the ensemble to be interrogated. Accordingly, this requires a {\it prior state update rule}. This {\it is} the belief updating alluded to in secs.~\ref{sec_states} and \ref{sec_noinfo} about the ensemble $\{S_a\}$.

It will not be necessary to specify these two update rules in detail; we just assume $O$ uses consistent ones. Specifically, given a posterior state update rule, we shall call $Q_i,Q_j\in\cq$
\begin{description}
\item[(maximally) independent] if, after having asked $Q_i$ to $S$ in the state of no information, the posterior probability $y_j=\f{1}{2}$. That is, if the answer to $Q_i$ relative to the state of no information tells $O$ `nothing' about the answer to $Q_j$.
\item[dependent] if, after having asked $Q_i$ to $S$ in the state of no information, the posterior probability $y_j\neq\f{1}{2}$. (If $y_j=0$ or $1$ they are maximally dependent.) That is, if the answer to $Q_i$ relative to the state of no information gives $O$ at least partial information about the answer to $Q_j$. 
\item[(maximally) compatible] if $O$ may know the answers to both $Q_i,Q_j$ simultaneously, i.e.\ if there exists a state in $\Sigma$ such that $y_i,y_j$ can be simultaneously $0$ or $1$.

\item[(maximally) complementary] if every state in $\Sigma$ which features $y_i=0,1$ necessarily implies $y_j=\f{1}{2}$. %In this case, the question $Q_{ij}$,  `are the answers to $Q_i,Q_j$ the same?', is an ill-defined question for $O$. Accordingly, we shall require, in addition, that in this case $O$ 
Notice that complementarity implies independence (but not vice versa).
\end{description}
(One can also define partial compatibility similarly \cite{Hoehn:2014uua}.) These relations shall be symmetric; e.g.\ $Q_i$ is independent of $Q_j$ if and only if $Q_j$ is independent of $Q_i$, etc.

We impose a final condition on the posterior state update rule: if $Q_i,Q_j$ are maximally compatible and independent then asking $Q_i$ shall not change $y_j$, i.e.\ $O$'s information about $Q_j$. %shall not gain or lose information about $Q_j$.

\subsection{Informational completeness}

The fundamental building blocks of the theories in the landscape which we are constructing are to be sets of pairwise independent questions. This will help to render the convoluted parametrization of a state by $\{y_i\}_{\forall\,Q_i\in\cq}$ more economical. Consider a set of pairwise independent questions $\cq_{M}:=\{Q_1,\ldots,Q_D\}$; it is called {\it maximal} if no question from $\cq\setminus\cq_M$ can be added to $\cq_M$ without destroying pairwise independence of its elements. We shall assume that any maximal $\cq_M$ is {\it informationally complete} in the sense that {\it all} $\{y_i\}_{\forall\,Q_i\in\cq}$ can be computed from the corresponding probabilities $\{y_i\}_{i=1}^D$ for all states in $\Sigma$. Any such $\cq_M$ features $D$ elements \cite{Hoehn:2014uua} such that $\Sigma$ becomes a $D$-dimensional convex set and states become vectors
\ba
\vec{y}=\left(\begin{array}{c}y_1 \\y_2 \\\vdots \\y_D\end{array}\right)\nn.
\ea

%Any convex set is defined by its extremal points \cite{Krein1940}. The extremal states in $\Sigma$ are special because they cannot be written as convex mixtures of other states, but all other states are convex mixtures of these. Since (finite) convex mixtures can be operationally understood in terms of (cascades of biased) coin flips, $O$ may prepare non-extremal states by applying cascades of coin flips to ensembles of extremal states. %The preparation procedure of a non-extremal state may not be unique since it may be a convex mixture of distinct sets of extremal states. 
%But, since the extremal states themselves cannot be prepared via coin flip cascades from other states, their preparation must have an unambiguous operational meaning. For this purpose, we wish any extremal state to be achievable by $O$ as the posterior state of an individual system in an interrogation. %We shall thus require that $O$ must be able to prepare any extremal state by engaging only what is available to him. 
%More specifically, we shall require that $O$ can prepare {\it any} extremal state from the state of no information in a {\it single shot interrogation}\footnote{In a single shot interrogation a single system $S$ is prepared in some state and subsequently exposed to questions (see \cite{Hoehn:2014uua} for more details).} by only asking questions from an informationally complete set $\cq_{M}$ and possibly letting the resulting state evolve in time. 

\subsection{Information measure}

Our focus is $O$'s acquisition of information, so we need to quantify $O$'s information about the systems. Since $Q_i\in\cq$ is binary, we quantify $O$'s information about $S_a$'s answer to it by a function $\alpha(y_i)$ with $0\leq\alpha(y_i)\leq 1$ \texttt{bit} and $\alpha(y)=0$ \texttt{bit} $\Leftrightarrow$ $y=\f{1}{2}$ and $\alpha(1)=\alpha(0)=1$ \texttt{bit}. $O$'s total information about a $S_a$ must be a function of the state; we make an additive ansatz
\ba
I(\vec{y}):=\sum_{i=1}^D\,\alpha(y_i)\label{infmeas}.
\ea
The quantum postulates will single out the specific function $\alpha$. %This will become crucial.

Consider a set $\{Q_1,\ldots,Q_n\}$ of mutually (maximally) complementary questions. It is clear that whenever $O$ has maximal information $\alpha(y_i)=1$ \texttt{bit} about $Q_i$ from this set, he must have $0$ \texttt{bits} of information about all other questions in the set. We require more generally that such a set cannot support more than $1$ \texttt{bit} of information, regardless of the state
\ba
\alpha(y_1)+\cdots+\alpha(y_n)\leq 1 \,\,\texttt{bit}\label{compstrong}
\ea
for otherwise $O$ could, for some states, {\it reduce} his total information about such a set by asking another question from it. These {\it complementarity inequalities} represent informational uncertainty relations that describe how the information gain about one question enforces an information loss about questions complementary to it (see also the state `collapse' in sec.~\ref{sec_indep}).

\subsection{Composite systems and (classical) rules of inference}\label{sec_comp}

$O$ must be able to tell a composite system apart into its constituents purely by means of the information accessible to him through interrogation and thus ultimately by means of the question sets. Let systems $S_{A},S_{B}$ have question sets $\cq_{A},\cq_{B}$. It is then natural to say that they define a composite system $S_{AB}$ if any $Q_a\in\cq_A$ is maximally compatible with any $Q_b\in\cq_B$ and if
\ba
\cq_{AB}=\cq_A\cup\cq_B\cup\tilde{\cq}_{AB},%\{Q_A*Q_B\big|\,\,Q_{A,B}\in\cq_{A,B},\,\,*\,\, \text{some logical connective}\}
\label{composite}
\ea
where $\tilde{\cq}_{AB}$ only contains composite questions which are iterative compositions, $Q_a\,*_{\tiny1}\,Q_b, Q_a\,*_2(Q_{a'}*_3Q_b), (Q_a*_4Q_b)*_5Q_{b'}, (Q_a*_6Q_b)*_7(Q_{a'}*_8Q_{b'}),\ldots$, via some logical connectives $*_1,*_2,*_3,\cdots$, of individual questions $Q_a,Q_{a'},\ldots\in\cq_A$ about $S_A$ and $Q_b,Q_{b'},\ldots\in\cq_B$ about $S_B$. This definition is extended recursively to composite systems with more than two subsystems.

Since $O$ can never test the truthfulness of statements about logical connectives of complementary questions through interrogations and since all propositions must have operational meaning, we shall permit $O$  to logically connect two (possibly composite) questions {\it directly} with some $*$ only if they are compatible. For the same reason, $O$ is allowed to apply classical rules of inference (in terms of Boolean logic) exclusively to sets of {mutually compatible} questions.

We stress that this definition of composite systems is distinct from the usual state tensor product rule in generalized probabilistic theories coming from local tomography \cite{Hardy:2001jk,masanes2011derivation,Dakic:2009bh}. In particular, this composition rule admits non-locally tomographic composites (see sec.~\ref{sec_3d}).

\subsection{Computing probabilities and questions as vectors}\label{sec_prob}

Thanks to informational completeness, the probability function $Y(Q|\vec{y})\in[0,1]$ that $Q=$`yes', given the state $\vec{y}$, exists for all $Q\in\cq$ and $\vec{y}\in\Sigma$. As shown in \cite{hw}, the exhibited structure yields
\ba
Y(Q|\vec{y})=Y(\vec{q}|\vec{y})=\f{1}{2}\left(\vec{q}\cdot(2\vec{y}-\vec{1})+1\right),\label{ansatz}
\ea
where $\vec{q}\in\mathbb{R}^D$ is a {\it question vector} encoding $Q\in\cq$ and $\vec{1}$ is a vector with each coefficient equal to $1$ in the basis corresponding to $\cq_M$. This equation gives rise to (part of) the Born rule.

Suppose $Q,Q'\in\cq$ were both encoded by the same $\vec{q}$. Then, by (\ref{ansatz}), they would be probabilistically indistinguishable and $O$ must view them as logically equivalent. $O$ is free to remove any such redundancy from his description of $\cq$ upon which every permissible question vector $\vec{q}$ will encode a {\it unique} $Q\in\cq$. Finally, for every $Q\in\cq$ there exists a state $\vec{y}_Q$ which is the updated posterior state of $S_a$ after $O$ received a `yes' answer to the single question $Q$ from $S_a$ in the (prior) state of no information. $O$ had $0$ \texttt{bits} of information before and $\vec{y}_Q$ encodes a single independent question answer, so we naturally require that it encodes $1$ independent \texttt{bit}. Hence, for every $Q\in\cq$ there exists $\vec{y}_Q\in\Sigma$ with $I(\vec{y}_Q)=1$ \texttt{bit} such that $Y(Q|\vec{y}_Q)=1$.\footnote{In quantum theory, the $\vec{y}_Q$ will only turn out to be pure states for a single qubit; e.g., for two qubits and $Q=$ `is the spin of qubit 1 up in $z$-direction?', represented by the rank-two projector $P_{z_1}=\f{1}{2}(\mathds{1}+\sigma_z\otimes\mathds{1}_{2\times2})$, $\vec{y}_Q$ corresponds to the mixed state $\rho_{z_1}=\f{1}{4}(\mathds{1}+\sigma_z\otimes\mathds{1}_{2\times2})$. Clearly, $\tr(P_{z_1}\rho_{z_1})=1$.}

%This concludes our review of the landscape of inference theories. 

%onlyTherefore we will allow the existence of a set of {\it time evolutions} denoted by $\ct$, where each element in $\ct$ is a function that maps one state in $\Sigma$ to another state in $\Sigma$ {\bf or should we say ''maps $\Sigma$ to itself``?}. We furthermore impose that the time evolution map only depends on the time interval and not on the particular point in time itself. Note however that the pair $(\cq,\Sigma)$ is required to be time-independent. 

\section{The quantum principles as rules constraining $O$'s information acquisition}\label{sec_axioms}

In the sequel, we consider the most elementary of information carriers. Within the introduced landscape of theories, we now establish rules on $O$'s acquisition of information that single out the quantum theory of a composite system $S_N$ of $N\in\mathbb{N}$ qubits, modelled in our language by a triple $(\cq_N,\Sigma_N,\ct_N)$. Effectively, these rules constitute a set of `coordinates' for quantum theory on this landscape. The rules are spelled out first colloquially, then mathematically and are motivated in more detail in \cite{Hoehn:2014uua,hw}. 

Empirically, the information accessible to an experimenter about (characteristic properties of) elementary systems is limited. For example, an experimenter may know one binary proposition about an electron (e.g., its spin in $x$-direction), but nothing fully independent of it (and similarly for a classical bit). We shall characterize a composition of $N$ elementary systems according to how much information is, in principle, simultaneously available to $O$. 

\begin{pr}\label{lim}{\bf(Limited Information)}
\emph{``The observer $O$ can acquire maximally $N\in\mathbb{N}$ {\it independent} \texttt{bits} of information about the system $S_N$ at any moment of time.''} \\
There exists a maximal set $Q_i$, $i=1,\ldots,N$, of $N$ mutually maximally independent and compatible questions in $\cq_N$. %and no set of mutually independent and compatible questions with more than $N$ elements exists.
\end{pr}
$O$ can thereby distinguish maximally $2^N$ states of $S_N$ in a single shot interrogation.

But, empirically, elementary systems admit more independent propositions than what -- due to the information limit -- they are able to answer at a time. This is Bohr's complementarity. The unanswered properties must be random (and so `in superposition') because the information limit makes it impossible to ascribe definite outcomes to them. For example, an experimenter may also inquire about the spin of the electron in $y$-direction. Yet doing so is at the total expense of his information about its spin in the $x$- and $z$-directions and subsequent such measurements have random outcomes. For the $N$ elementary systems, we assert the existence of complementarity.
\begin{pr}\label{unlim}{\bf(Complementarity)}
\emph{``The observer $O$ can always get up to $N$ \emph{new} independent \texttt{bits} of information about the system $S_N$. But whenever $O$ asks $S_N$ a new question, he experiences no net loss in his total amount of information about $S_N$.''}\\
There exists another maximal set $Q_i'$, $i=1,\ldots,N$, of $N$ mutually maximally independent and compatible questions in $\cq_N$ such that $Q'_i,Q_i$ are maximally complementary and $Q'_i,Q_{j\neq i}$ are maximally compatible.
\end{pr}
The peculiar mathematical form of rule \ref{unlim} becomes intuitive upon recalling that $S_N$ is a {\it composite} system such that complementarity should exist {\it per} elementary system \cite{Hoehn:2014uua}. 

Rules \ref{lim} and \ref{unlim} are conceptually inspired by (non-technical) proposals made by Rovelli \cite{Rovelli:1995fv} and Zeilinger and Brukner \cite{zeilinger1999foundational,Brukner:2002kx}. These rules say nothing about what happens in-between interrogations. Naturally, we demand $O$ not to gain or lose information {\it without} asking questions.

\begin{pr}\label{pres}{\bf(Information Preservation)}
\emph{``The total amount of information $O$ has about (an otherwise non-interacting) $S_N$ is preserved in-between interrogations."}\\
$I(\vec{y})$ is \emph{constant} in time in-between interrogations for (an otherwise non-interacting) $S_N$.
\end{pr}
Hence, $O$'s total information $I(\vec{y})$ is a `conserved charge' of any time evolution $T_{\Delta t}\in\ct_N$.

%We shall now stipulate that the maximal set of time evolutions which can be made consistent with the other principles is physically realizable. This renders the set of legal time evolutions and the state spaces interdependent because any time evolution must map a legal state to another one, such that neither can exist independently of the other.

The more interactions to which $O$ may subject $S_N$ are available, the more ways in which any state may, in principle, change in time and thus the more `interesting' $O$'s world. We therefore demand that {\it any} time evolution is physically realizable as long as it is consistent with the other rules. (Since $\Sigma_N,\ct_N$ are interdependent, this is distinct from `maximizing the number' of states.) 

\begin{pr}\label{time}{\bf(Time Evolution)}\footnote{If we did not require this `maximality' of $\ct_N$, we would still ultimately obtain a linear, unitary evolution, but not necessarily the full unitary group. This is the sole reason for demanding `maximality'. Note that principles \ref{pres} and \ref{time} are {\it not} equivalent to the axiom of `continuous reversibility' of generalized probabilistic theories \cite{Hardy:2001jk,masanes2011derivation,Dakic:2009bh}.} 
\emph{``$O$'s `catalogue of knowledge' about $S_N$ evolves \emph{continuously} in time in-between interrogations and every consistent such evolution is physically realizable.''}\\
%Every function $T_{\Delta t}$ that maps any given state $\vec{y}$ 
%\emph{continuously} in $\Delta t$ is contained in $\ct_N$ if it is consistent with the structure of the theory landscape and $T_{\Delta t}(\vec{y})$ is compatible with principles \ref{lim}-\ref{pres} $\forall\,\vec{y}\in\Sigma_N$.
%A transformation $T_{\Delta t}$ on states is contained in $\ct_N$ if and only if $T_{\Delta t}(\vec{y})$ is \emph{continuous} in $\Delta t$ and compatible with principles \ref{lim}-\ref{pres} (and the structure of the theory landscape) for any fixed state $\vec{y}$.
$\ct_N$ is the maximal set of transformations $T_{\Delta t}$ on states such that, for any \emph{fixed} state $\vec{y}$, $T_{\Delta t}(\vec{y})$ is \emph{continuous} in $\Delta t$ and compatible with principles \ref{lim}-\ref{pres} (and the structure of the theory landscape).
%$T_{\Delta t}(\vec{y})$ is \emph{continuous} in $\Delta t$ $\forall\,T_{\Delta t}\in\ct_N$ and $\vec{y}\in\Sigma_N$. Any transformation on states is contained in...
\end{pr}

We shall also allow $O$ to ask {\it any} question to $S_N$ which `makes (probabilistic) sense'.

\begin{pr}\label{Q}{\bf(Question Unrestrictedness)}\footnote{Without principle \ref{Q}, we would still obtain the structure of an informationally complete set $\cq_{M_N}$, finding that it encodes a basis of projective Pauli operator measurements \cite{hw}; principle \ref{Q} legalizes {\it all} such measurements.}
\emph{``Every question which yields legitimate probabilities for every way of preparing $S_N$ is physically realizable by $O$.''}\\
Every question vector $\vec{q}\in\mathbb{R}^{D_N}$ which satisfies $Y(\vec{q}|\vec{y})\in[0,1]$ $\forall\,\vec{y}\in\Sigma_N$ and for which there exists $\vec{y}_Q\in\Sigma_N$ with $I(\vec{y}_Q)=1$ \texttt{bit} such that $Y(\vec{q}|\vec{y}_Q)=1$ corresponds to a $Q\in\cq_N$.
 \end{pr}

These five rules turn out to leave two solutions for the triple $(\cq_N,\Sigma_N,\ct_N)$. Remarkably, they cannot distinguish between complex and real numbers. Namely, the two solutions are qubit and rebit quantum theory, i.e.\ two-level systems over real Hilbert spaces \cite{Hoehn:2014uua,hw}. Since the latter is both mathematically and physically a subcase of the former, these five rules can be regarded as sufficient. However, if one also wishes to discriminate rebits operationally, then an extra rule, adapted from \cite{Hardy:2001jk,Dakic:2009bh,masanes2011derivation} and imposed {\it solely} for this purpose (it is partially redundant), succeeds.

\begin{pr}\label{loc}{\bf(Tomographic Locality)}
\emph{``$O$ can determine the state of the composite system $S_N$ by interrogating only its subsystems."}
\end{pr}

As shown in \cite{Hoehn:2014uua,hw}, rules \ref{lim}--\ref{loc} are equivalent to the textbook axioms. More precisely:
\begin{claim}
The only solution to rules \ref{lim}--\ref{loc} is qubit quantum theory where
\begin{itemize}
\item $\Sigma_N\simeq\text{convex hull of }\mathbb{CP}^{2^N-1}$ is the space of $2^N\times2^N$ density matrices over $\mathbb{C}^{2^N}$,
\item states evolve unitarily according to $\ct_N\simeq\rm{PSU}(2^N)$ and the equation describing the state dynamics is (equivalent to) the von Neumann evolution equation,
\item $\cq_{N}\simeq\mathbb{CP}^{2^N-1}$ is (isomorphic to) the set of projective measurements onto the $+1$ eigenspaces of $N$-qubit Pauli operators\footnote{A Hermitian operator on $\mathbb{C}^{2^N}$ is a Pauli operator iff it has two eigenvalues $\pm1$ of equal multiplicity.} and the probability for $Q\in\cq_N$ to be answered with `yes' in some state is given by the Born rule for projective measurements.
\end{itemize}
\end{claim}

\section{Synopsis of the reconstruction steps and key results}\label{sec_reconst}

Since this gives rise to a {\it constructive} derivation of the explicit architecture of qubit quantum theory, it involves a large number of individual steps compared to the rather abstract reconstructions \cite{Hardy:2001jk,Dakic:2009bh,masanes2011derivation,chiribella2011informational,Barnum,M2,2008arXiv0805.2770G,Fuchs}. However, this is also rewarding as it offers novel informational explanations for typical features of quantum theory and so many reconstruction steps are actually quite instructive. We now provide a summary of key results and reconstruction steps from \cite{Hoehn:2014uua,hw} (to which we refer for technical details) needed for proving the claim of the previous section.  
%For the details and proofs of the reconstruction, we refer the reader to \cite{Hoehn:2014uua,hw}.

\subsection{Logical connectives for building informationally complete sets}\label{sec_connect}

The first task is to build informationally complete sets $\cq_{M_N}$ \cite{Hoehn:2014uua}. The conjunction of rules \ref{lim} and \ref{unlim} implies that $\cq_{M_1}=\{Q_1,Q_2,\ldots,Q_{D_1}\}$ for a single elementary system must be a maximal mutually complementary set with $D_1\geq2$. We changed notation slightly, labelling complementary questions by numbers, not primes. Of course, in quantum theory, $D_1=3$; the more involved $N=2$ case will entail this. The structure (\ref{composite}) of a composite system implies that $\cq_{M_2}$ should contain individual questions about its subsystems. Continuing with a slight change of notation, we denote $\cq_{M_1}$ for system 1 by $\{Q_1,Q_2,\ldots,Q_{D_1}\}$ and for system 2 with a prime by $\{Q'_1,Q'_2,\ldots,Q'_{D_1}\}$. %, represented as {\it vertices} in the question graph of figure \ref{}.
Apart from these individual questions, $\cq_{M_2}$ should contain composite questions $Q_i*Q'_j$ for some connective $*$. Pairwise independence of $\cq_{M_2}$ enforces that $*$ must satisfy the following truth table, where `yes'$=1$ and `no'$=0$ ($Q_i,Q'_j$ are compatible) \cite{Hoehn:2014uua}:
\begin{eqnarray}
%\begin{center}
  \begin{tabular}{ ||c | c || c|| }
    \hline
    $Q_{i}$ & $Q'_{j}$ & $Q_i*Q'_j$ \\ \hline\hline
    0 & 1 & a \\ \hline
    1 & 0 & a \\ \hline
    1&1&b\\\hline
    0&0&b\\\hline
  \end{tabular}\q\q\q\q \q\q a\neq b\,\q\q a,b\in\{0,1\}.\label{truth}
%\end{center}
\end{eqnarray}
Hence, $*$ is either the XNOR $\leftrightarrow$ (for $a=0$, $b=1$), or its negation, the XOR $\oplus$ (for $a=1$, $b=0$). Up to an overall negation $\neg$, the two connectives are logically equivalent and so we henceforth make the convention to only build up composite questions (for informationally complete sets) using the XNOR. The composite question $Q_{ij}:=Q_i\leftrightarrow Q'_j$ is a `correlation question', representing ``are the answers to $Q_i,Q'_j$ the same?." Ultimately, in quantum theory, $\leftrightarrow$ will turn out to correspond to the tensor product $\otimes$ in $\sigma_i\otimes\sigma_j$ where $\sigma_i$ is a Pauli matrix; $Q_{ij}$ will then correspond to ``are the spins of qubit 1 in $i$- and of qubit 2 in $j$-direction correlated?."

\subsection{Question graphs, independence and compatibility for $N=2$ and entanglement}
It is convenient to represent questions graphically: individual questions are represented as {\it vertices} and bipartite correlation questions as {\it edges} between them. For instance, we may have
\begin{eqnarray}
\psfrag{q1}{\hspace*{-.3cm}\tiny system 1}
\psfrag{q2}{\tiny system 2}
\psfrag{Q1}{$Q_1$}
\psfrag{Q2}{$Q_2$}
\psfrag{Q3}{$Q_3$}
\psfrag{QD}{$Q_{D_1}$}
\psfrag{P1}{$Q'_1$}
\psfrag{P2}{$Q'_2$}
\psfrag{P3}{$Q'_3$}
\psfrag{PD}{$Q'_{D_1}$}
\psfrag{q11}{$Q_{11}$}
\psfrag{q31}{$Q_{31}$}
\psfrag{q22}{$Q_{22}$}
\psfrag{q23}{\hspace*{-.2cm}$Q_{23}$}
\psfrag{qdd}{$Q_{D_1D_1}$}
\psfrag{d}{$\vdots$}
{\includegraphics[scale=.2]{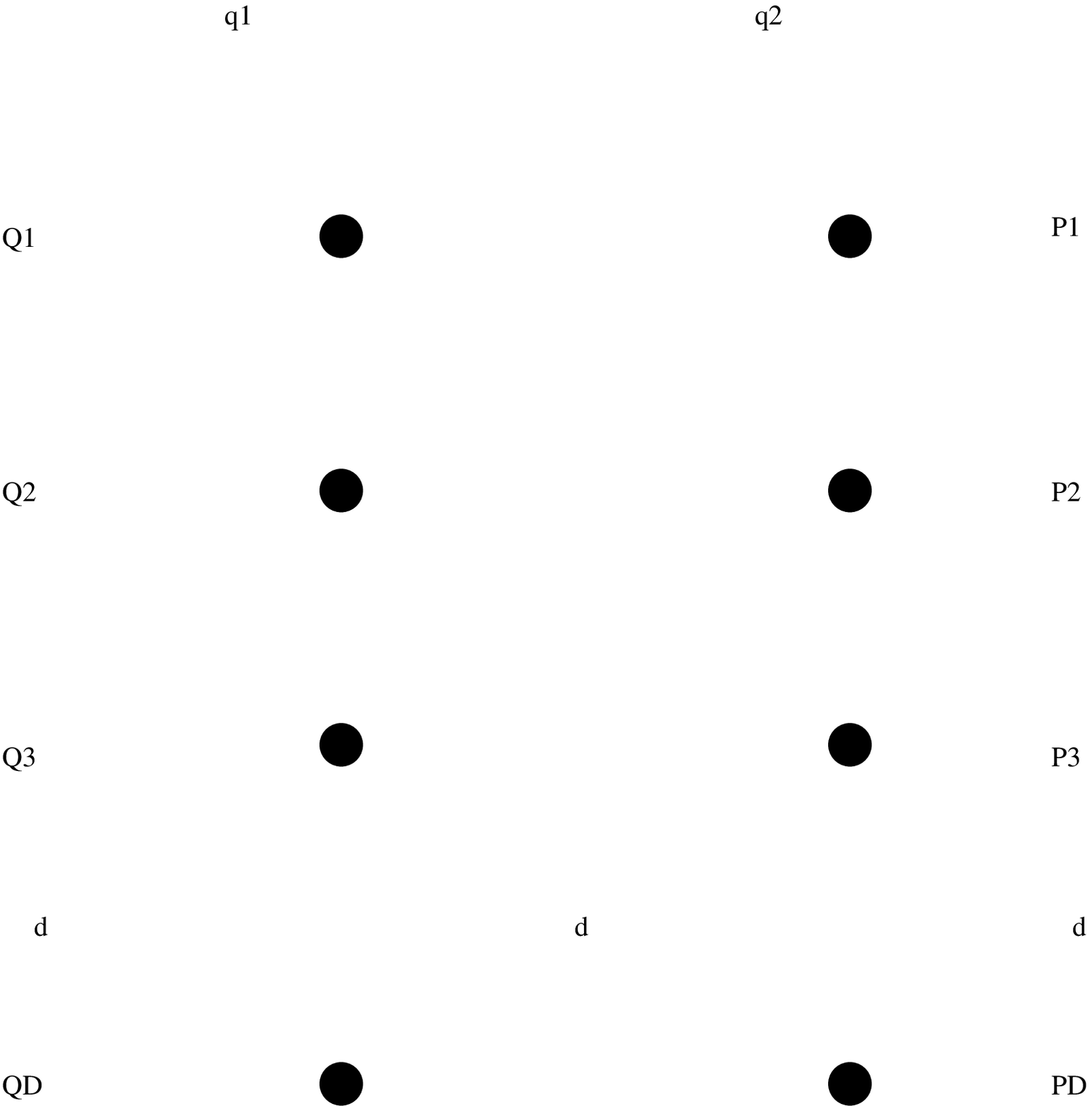}}\q\q\q\q\q\q{\includegraphics[scale=.2]{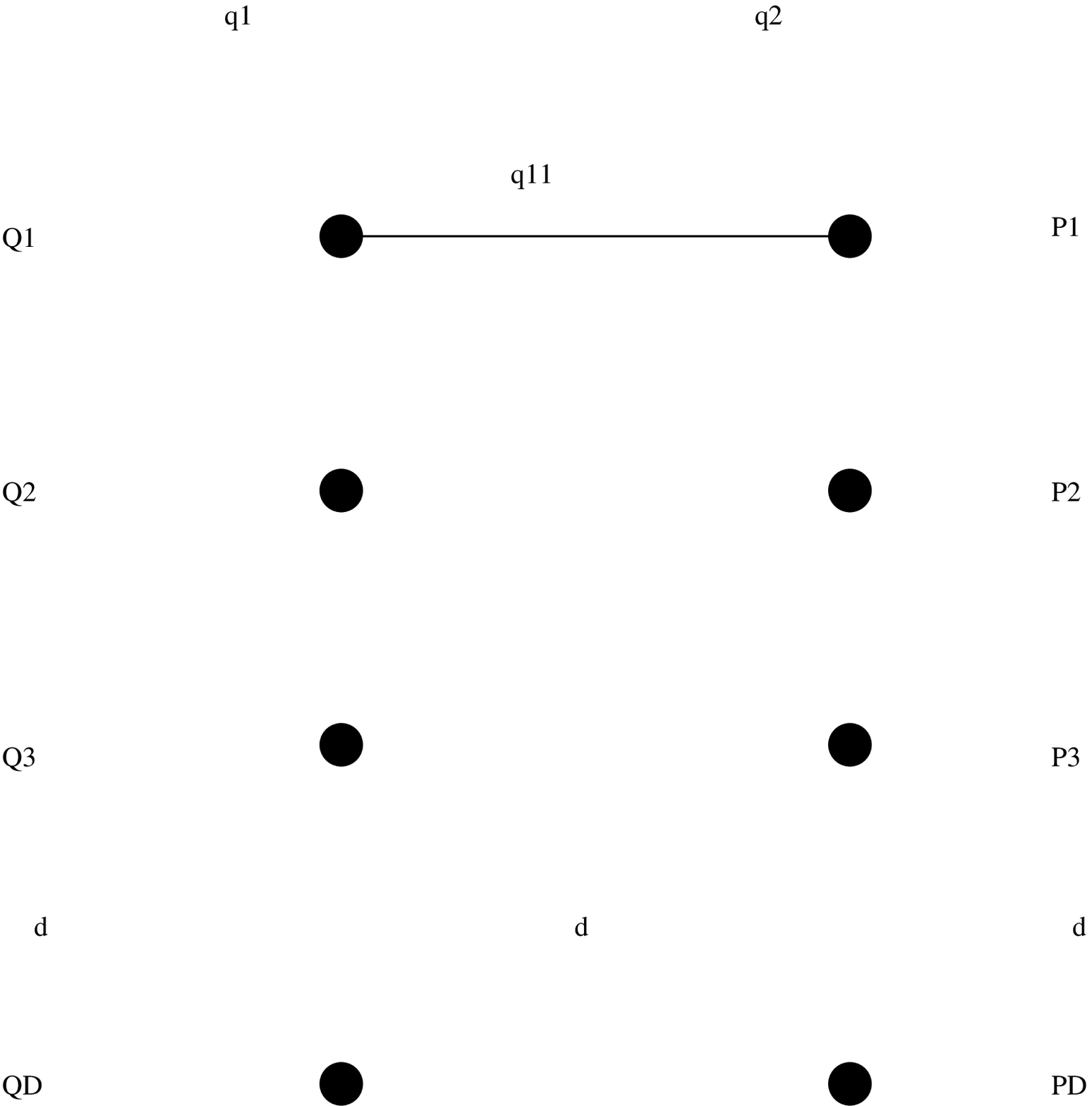}}\q\q\q\q\q\q{\includegraphics[scale=.2]{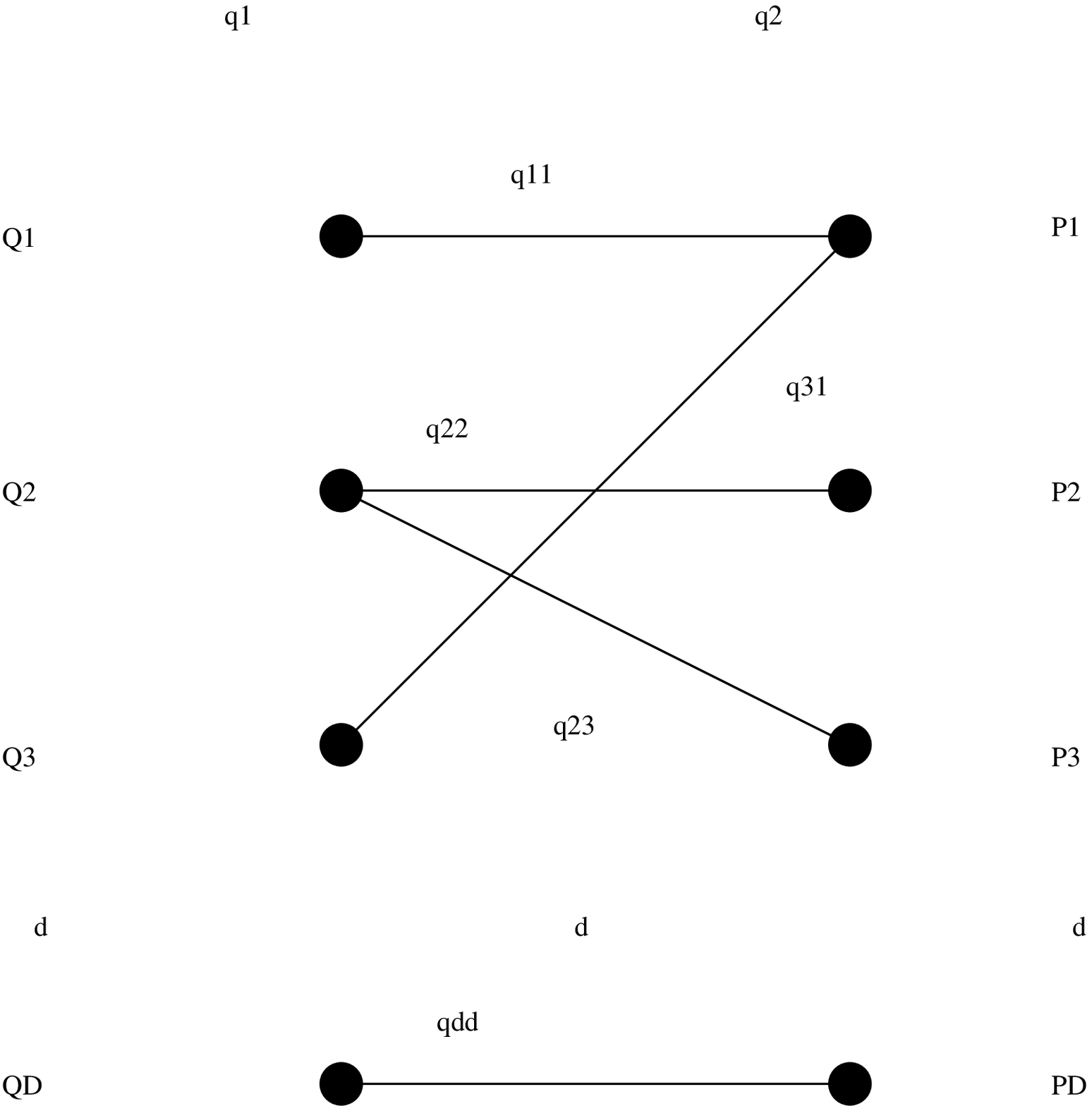}}.\notag
\end{eqnarray}
Since $O$ is only allowed to connect compatible questions logically, there can be no edge between individual questions of the {\it same} system.

Using only rules \ref{lim} and \ref{unlim} and logical arguments, the following result is proven in \cite{Hoehn:2014uua}:
\begin{lemma}\label{lem1}
$Q_i,Q'_j,Q_{ij}$ are pairwise independent for all $i,j=1,\dots,D_1$ and will thus be part of an informationally complete set $\cq_{M_2}$. Furthermore:
\begin{enumerate}
\item $Q_i$ is compatible with $Q_{ij}$, $\forall\,j=1,\ldots,D_1$ and complementary to $Q_{kj}$, $\forall\,k\neq i$ and $\forall\,j=1,\ldots,D_1$. That is, graphically, an individual question $Q_i$ is compatible with a correlation question $Q_{ij}$ if and only if its corresponding vertex is a vertex of the edge corresponding to $Q_{ij}$. By symmetry, the analogous result holds for $Q'_j$.
\item $Q_{ij}$ and $Q_{kl}$ are compatible if and only if $i\neq k$ and $j\neq l$. That is, graphically, $Q_{ij}$ and $Q_{kl}$ are compatible if their corresponding edges do \emph{not} intersect in a vertex and complementary if they intersect in one vertex.
\end{enumerate}
\end{lemma}

For example, $Q_1$ in the third question graph above is compatible with $Q_{11}$ and complementary to $Q_{22}$, while $Q_{11}$ and $Q_{22}$ are compatible and $Q_{11}$ and $Q_{31}$ are complementary.

This lemma has a striking consequence: it implies {\it entanglement}. Indeed, since, e.g., $Q_{11}$ and $Q_{22}$ are independent and compatible, $O$ may spend his maximally accessible amount of $N=2$ {\it independent} \texttt{bits} of information (rule \ref{lim}) over correlation questions only. Since nonintersecting edges do not share a common vertex, the lemma implies that no individual question is simultaneously compatible with two correlation questions that are compatible. Hence, when knowing the answers to $Q_{11},Q_{22}$, $O$ will be entirely ignorant about the individual questions; $O$ has then maximal information about $S_2$, but purely composite information. This is entanglement in the very sense of Schr\"odinger (\emph{``...the best possible knowledge of a whole does not necessarily include the best possible knowledge of all its parts..."} \cite{schrod}). For example, in quantum theory, a state with $Q_{11}=Q_{22}=$ `yes' will coincide with a Bell state having the spins of qubits 1 and 2 correlated in $x$- and $y$-direction (and anti-correlated in $z$-direction). Of course, there is nothing special about $Q_{11},Q_{22}$ and the argument works similarly for other composite question pairs and can be extended also to states with non-maximal entanglement (see \cite{Hoehn:2014uua} for details).

For systems with limited information content, {\it entanglement is therefore a direct consequence of complementarity}; without it there would be no independent and compatible composite questions sufficient to saturate the information limit \cite{Hoehn:2014uua}. For instance, two classical bits satisfy rule \ref{lim} as well, but admit no complementarity so that $\cq_{M_2}^{\rm cbit}=\{Q_1,Q'_1,Q_{11}\}$ and the maximum amount of $N=2$ {\it independent} \text{bits} can{\it not} be spent on composite questions only.

 \begin{wrapfigure}{o}{0.23\textwidth}
  \vspace{-15pt}
  \begin{center} 
\psfrag{a}{\hspace{-.1cm}$S_A$}
 \psfrag{b}{$S_B$}
 \psfrag{c}{$S_C$}
\includegraphics[scale=.2]{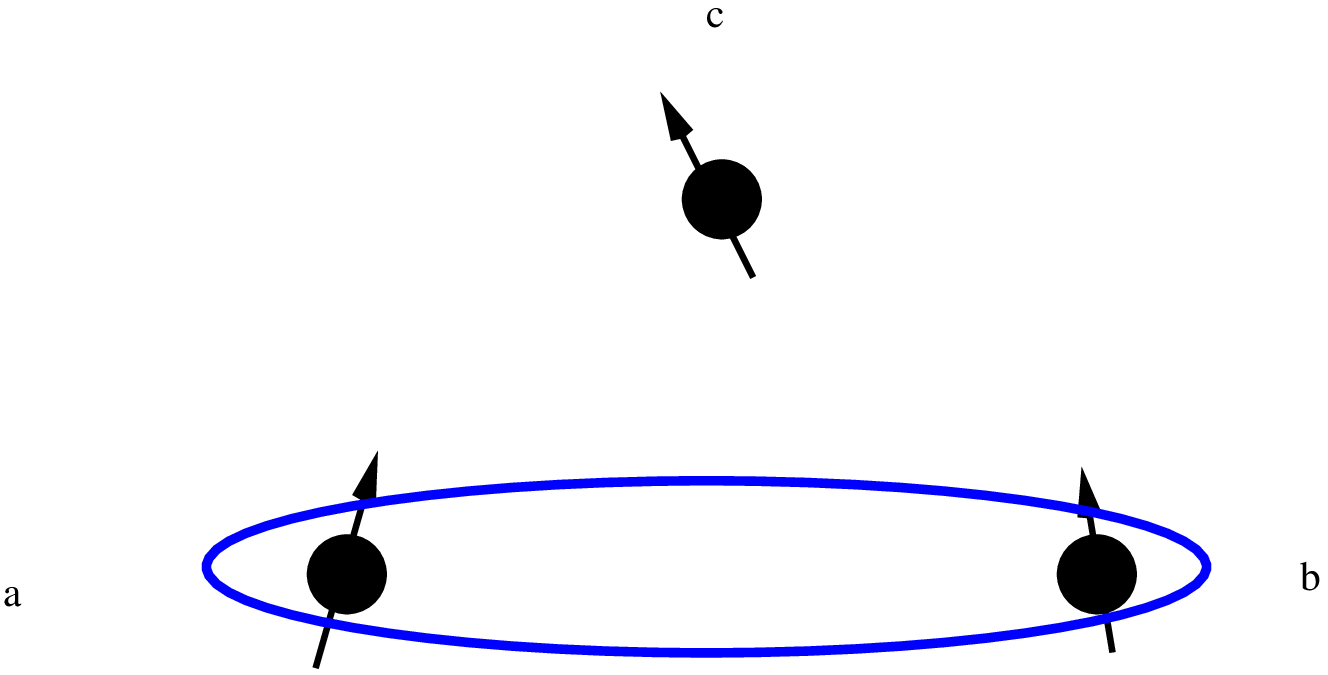}
\end{center}  \vspace{-20pt}
\end{wrapfigure}
We also note that rules \ref{lim} and \ref{unlim} offer a simple, intuitive explanation for {\it monogamy of entanglement}.
Consider, for a moment $N=3$ elementary systems $S_A,S_B,S_C$, and suppose $S_A$ and $S_B$ are maximally entangled (say, because $O$ received the answer $Q_{11}=Q_{22}=$ `yes' from $S_{AB}$). Noting that $S_{AB}$ is a composite bipartite system inside the tripartite $S_{ABC}$, $O$ has then already spent his maximal amount of information of $N=2$ {\it independent} \texttt{bits} which he may know about $S_{AB}$ and can therefore not know anything else that is independent, incl.\ non-trivial correlations with $S_C$, about the pair. To saturate the $N=3$ {\it independent} \text{bit} limit for the tripartite system $S_{ABC}$, he may then only inquire individual information about $S_C$. This is {\it monogamy} in its extreme form: the maximally entangled pair $S_{AB}$ can{\it not} be entangled with any other system $S_C$. This heuristic argument can be made rigorous in terms of the compatibility and independence structure of questions for $N\geq3$ and can be extended to the non-extremal case using informational {\it monogamy inequalities} \cite{Hoehn:2014uua}.

\subsection{A logical explanation for the three-dimensionality of the Bloch ball}\label{sec_3d}

A key result of the reconstruction, proven in \cite{Hoehn:2014uua}, is the following. Since its proof is instructive and representative for this approach, we shall rephrase it here.
\begin{theorem}\label{thm_3d}
$D_1=2$ or $3$.
\end{theorem}

\begin{proof}
Consider the $N=2$ case. Lemma \ref{lem1} implies that any maximal set of pairwise compatible correlation questions has $D_1$ elements. Indeed, there are maximally $D_1$ non-intersecting edges between the $D_1$ vertices of system 1 and the $D_1$ vertices of system 2; e.g., the $D_1$ `diagonal' $Q_{ii}$
\begin{eqnarray}
\psfrag{d}{$\vdots$}
\psfrag{Q11}{$Q_{11}$}
\psfrag{Q22}{$Q_{22}$}
\psfrag{QDD}{$Q_{D_1D_1}$}
\psfrag{Q33}{$Q_{33}$}
{\includegraphics[scale=.2]{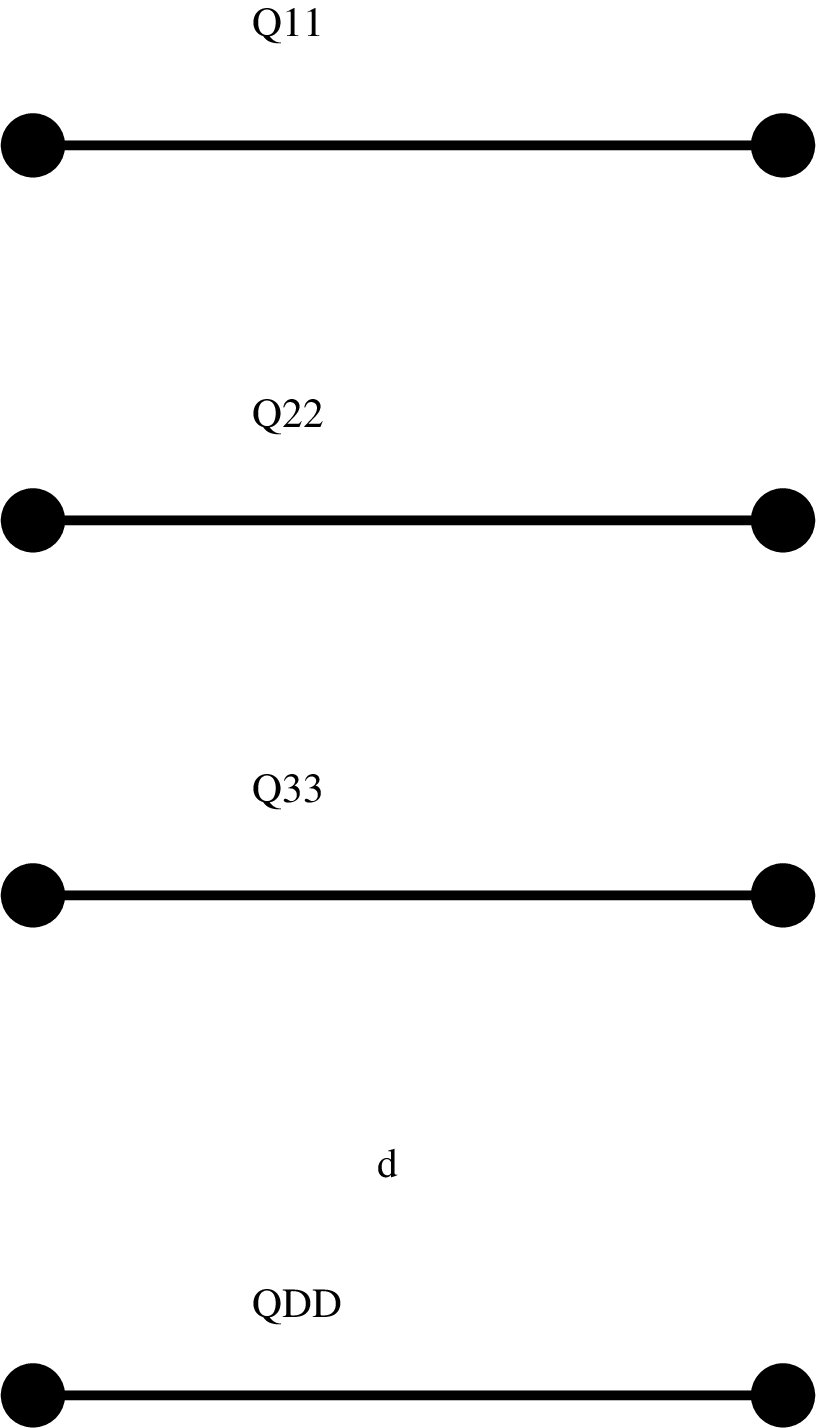}}\notag
\end{eqnarray}   
are pairwise independent and compatible. The constraints on the posterior state update rule in section \ref{sec_indep} entail that they are also mutually compatible (Specker's principle) \cite{Hoehn:2014uua} such that $O$ may simultaneously know the answers to all $D_1$ $Q_{ii}$. Since $O$ may not know more than $N=2$ independent \texttt{bits} (rule \ref{lim}), the $D_1$ $Q_{ii}$ can{\it not} be mutually independent if $D_1>2$. Thus, assuming the $Q_{ii}$ are of equivalent status, the answers to any pair of them, say $Q_{11},Q_{22}$, must imply the answers to all others, say $Q_{ii}$, $i=3,\ldots,D_1$. Hence, $Q_{jj}=Q_{11}*Q_{22}$, $j\neq1,2$, for a connective $*$ that preserves pairwise independence of $Q_{11},Q_{22},Q_{jj}$. Reasoning as in (\ref{truth}) implies that either
\ba
Q_{jj}=Q_{11}\leftrightarrow Q_{22},\q\q\q\text{or}\q\q\q Q_{jj}=\neg(Q_{11}\leftrightarrow Q_{22}),\q\q\q j=3,\ldots,D_1\label{3d}
\ea
so that for $D_1>3$ $Q_{jj}$, $j=3,\ldots,D_1$, can{\it not} be pairwise independent. Arguing identically for all other sets of $D_1$ pairwise independent and compatible $Q_{ij}$, we conclude that $D_1\leq3$.
\end{proof}

This theorem has several crucial repercussions. We may already suggestively call $D_1=2$ and $D_1=3$ the `rebit' (two-level systems over {\it real} Hilbert spaces) and `qubit' case, respectively. Reasoning as in (\ref{3d}) shows that the $Q_{ij}$ are logically closed under $\leftrightarrow$; as demonstrated in \cite{Hoehn:2014uua}:
\begin{theorem}\label{thm_qubit}
If $D_1=3$ then $\cq_{M_2}:=\{Q_i,Q'_j,Q_{ij}\}_{i,j=1,2,3}$ is logically closed under $\leftrightarrow$ and thus constitutes an informationally complete set for $N=2$ with $D_2=15$. 

If $D_1=2$ then $\cq_{M_2}=\{Q_i,Q_j',Q_{ij},Q_{11}\leftrightarrow Q_{22}\}_{i,j=1,2}$ is logically closed under $\leftrightarrow$ and thus constitutes an informationally complete set for $N=2$ with $D_2=9$. Furthermore, $Q_{11}\leftrightarrow Q_{22}$ is complementary to the individual questions $Q_i,Q'_j$, $i,j=1,2$.
\end{theorem}

Indeed, $D_2=9,15$ are the correct numbers of degrees of freedom for $N=2$ rebits and qubits, respectively. However, since the composite question $Q_{11}\leftrightarrow Q_{22}$ is complementary to {\it all} individual questions in the rebit case (this is {\it not} true in the qubit case!), it is impossible for $O$ to do ensemble state tomography by asking only individual questions $Q_i,Q'_j$, thereby violating rule \ref{loc}. We are left with the qubit case and shall henceforth ignore rebits (for rebits see \cite{Hoehn:2014uua}).

\subsection{Ruling out local hidden variables and the correlation structure for $N=2$}

Using (\ref{3d}) and repeating the argument leading to it for `non-diagonal' $Q_{ij}$ shows that either
\ba
Q_{11}\leftrightarrow Q_{22}=Q_{12}\leftrightarrow Q_{21},\q\q\q\text{or}\q\q\q Q_{11}\leftrightarrow Q_{22}=\neg(Q_{12}\leftrightarrow Q_{21}).\label{qbit6}
\ea
The first case (without relative negation) is the case of {\it classical} logic and compatible with {\it local} hidden variables for the individual questions $Q_i,Q'_j$. Namely, note that $Q_{11}\leftrightarrow Q_{22}=Q_{12}\leftrightarrow Q_{21}$ can be rewritten in terms of the individuals as
\ba
(Q_1\leftrightarrow Q_1')\leftrightarrow(Q_2\leftrightarrow Q_2')=(Q_1\leftrightarrow Q_2')\leftrightarrow(Q_2\leftrightarrow Q_1').\label{noclassfo}
\ea
Suppose for a moment that $Q_1,Q'_1,Q_2,Q'_2$ had simultaneous definite values (although not accessible to $O$). It is easy to convince oneself that any distribution of simultaneous truth values over the $Q_i,Q'_j$ satisfies (\ref{noclassfo}) \cite{Hoehn:2014uua}. In fact, (\ref{noclassfo}) is a {\it classical logical identity} and can be argued to follow from classical rules of inference \cite{Hoehn:2014uua}. However, it involves complementary individual questions, thereby violating our premise from section \ref{sec_comp} that $O$ may apply classical rules of inference exclusively to mutually compatible questions. This classical case is thus ruled out.

One can check that the second case, $Q_{11}\leftrightarrow Q_{22}=\neg(Q_{12}\leftrightarrow Q_{21})$, does {\it not} admit a {\it local} hidden variable interpretation, but is consistent with the structure of the theory landscape and rules \cite{Hoehn:2014uua}. Since one of the two cases (\ref{qbit6}) {\it must} be true, we conclude that this second case holds. In fact, for {\it any} complementary pairs $Q,Q'$ and $Q'',Q'''$ such that both $Q$ and $Q'$ are compatible with both $Q'',Q'''$, one finds similarly \cite{Hoehn:2014uua}
\ba
(Q\leftrightarrow Q'')\leftrightarrow (Q'\leftrightarrow Q''')=\neg\left((Q\leftrightarrow Q''')\leftrightarrow(Q'\leftrightarrow Q'')\right).\label{qbit8}
\ea 
This precludes to reason classically about the distribution of truth values over $O$'s questions.

(\ref{qbit8}) permits us to unravel the complete correlation structure for $\cq_{M_2}$. In fact, it turns out that there are two distinct representations of this correlation structure: one corresponding to quantum theory in its standard representation, the other to its `mirror' representation, related by a {\it passive} (not a physical) transformation, reassigning $Q_1\mapsto\neg Q_1$ (in quantum theory tantamount to a partial transpose on qubit 1) \cite{Hoehn:2014uua}. The two distinct representations turn out to be physically equivalent and so a convention has to be made. Choosing the `standard' case and using (\ref{qbit8}), one finds that the compatibility and correlation structure of $\cq_{M_2}$ can be represented graphically as in fig.~\ref{fig_corr}. 
For $Q,Q',Q''$ compatible, we shall henceforth distinguish between
\begin{description}
\item[even correlation:] if $Q=Q'\leftrightarrow Q''$, and
\item[odd correlation:] if $Q=\neg(Q'\leftrightarrow Q'')$.
\end{description}
\begin{figure}[hbt!]
\begin{center}
\psfrag{+}{\scriptsize$+$}
\psfrag{-}{\scriptsize\hspace*{-.1cm}$-$}
\psfrag{1}{\scriptsize$Q_1$}
\psfrag{2}{\scriptsize\hspace*{-.1cm}$Q_2$}
\psfrag{3}{\scriptsize$Q_3$}
\psfrag{1p}{\scriptsize$Q'_1$}
\psfrag{2p}{\scriptsize\hspace*{-.1cm}$Q'_2$}
\psfrag{3p}{\scriptsize$Q'_3$}
\psfrag{11}{\scriptsize$Q_{11}$}
\psfrag{22}{\scriptsize$Q_{22}$}
\psfrag{12}{\scriptsize\hspace{-.1cm}$Q_{12}$}
\psfrag{33}{\scriptsize$Q_{33}$}
\psfrag{13}{\scriptsize$Q_{13}$}
\psfrag{21}{\scriptsize$Q_{21}$}
\psfrag{23}{\scriptsize$Q_{23}$}
\psfrag{31}{\scriptsize\hspace{-.1cm}$Q_{31}$}
\psfrag{32}{\scriptsize$Q_{32}$}
\psfrag{i}{\hspace*{-.45cm}\tiny identify}
{\includegraphics[scale=.2]{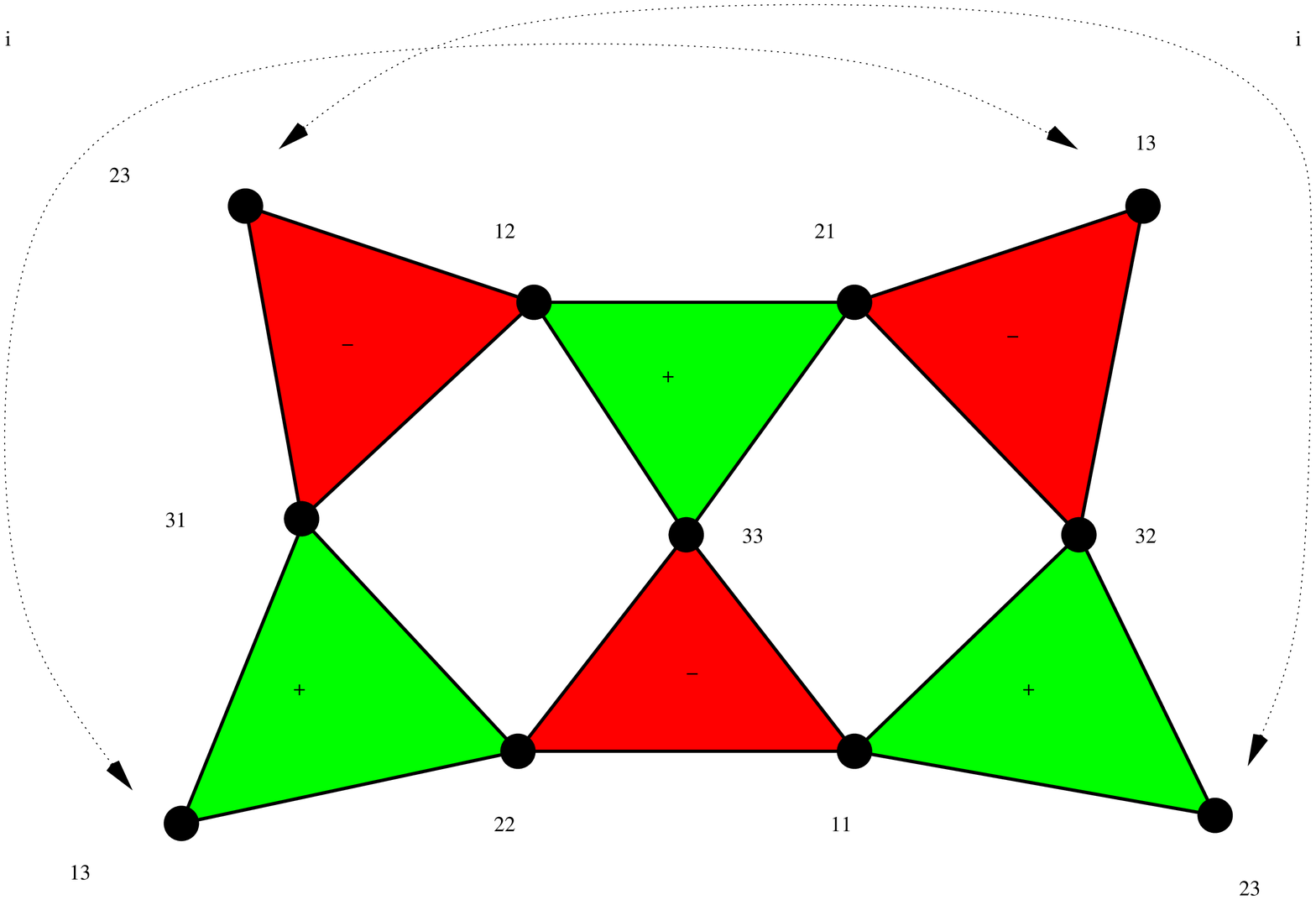}}\q\q\q%\\\vspace*{1cm}
{\includegraphics[scale=.2]{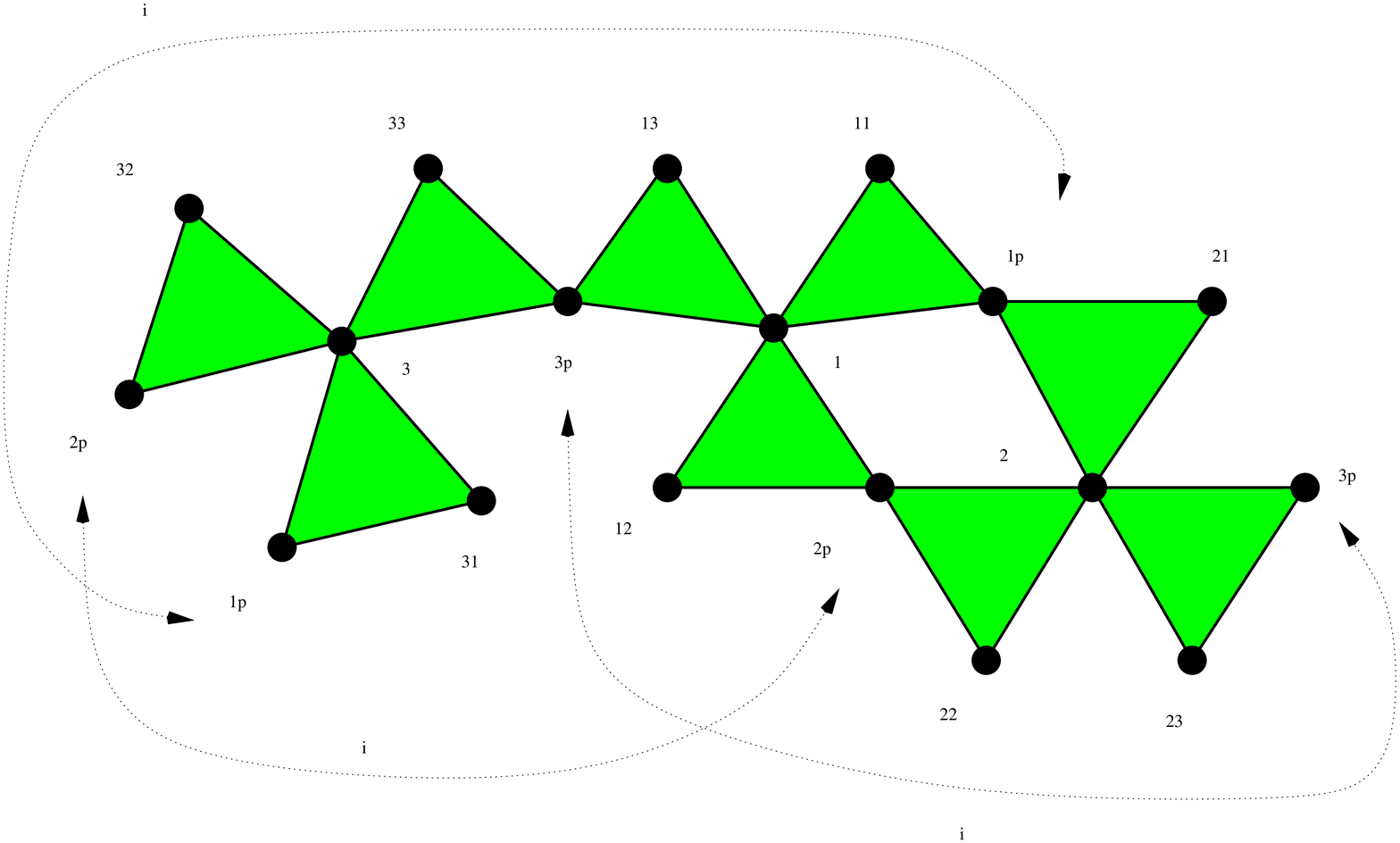}}
\caption{\small The compatibility and correlation structure of the informationally complete set $\cq_{M_2}$ for the $N=2$ qubit case. Two questions are compatible if connected by a triangle edge and complementary otherwise. Red and green triangles denote odd and even correlation, respectively; e.g., $Q_{33}=\neg(Q_{11}\leftrightarrow Q_{22})=Q_{12}\leftrightarrow Q_{21}$. (Taken from \cite{Hoehn:2014uua}.)}\label{fig_corr}
\end{center}
\end{figure}

One can easily check that quantum theory satisfies this correlation structure for projective spin measurements if one replaces $i=1,2,3$ by $x,y,z$. For instance, $Q_{11}=Q_{22}=$ `yes' implies, by fig.~\ref{fig_corr}, the dependent $Q_{33}=$ `no'. In quantum theory, this corresponds to the (unnormalized) Bell state with spin correlation in $x$- and $y$-direction and anti-correlated spins in $z$-direction
\ba
|x_+x_+\rangle-|x_-x_-\rangle=-i|y_+y_+\rangle+i|y_-y_-\rangle=|z_+z_-\rangle+|z_-z_+\rangle.\nn
\ea

\subsection{Compatibility, independence and informational completeness for arbitrary $N$}

Consider $N$ elementary systems in the `qubit' ($D_1=3$) case and the XNOR conjunction 
\ba
 Q_{\mu_1\mu_2\cdots \mu_N}:=Q_{\mu_1}\leftrightarrow Q_{\mu_2}\leftrightarrow\cdots\leftrightarrow Q_{\mu_N}\label{ngbit}
 \ea
of individual questions, where $\mu_a=0,1,2,3$ and $Q_{0}:=$`yes'. The conjunction yields `yes' and `no' if an even and odd number of $Q_{\mu_a}=$ `no', respectively, and thus does {\it not} represent ``are the answers to all $Q_{\mu_a}$ the same?.'' As shown in \cite{Hoehn:2014uua}, these conjunctions are informationally complete:

\begin{theorem} {\bf(Qubits)}\label{thm_nqbit}
The $4^N-1$ questions\footnote{We deduct the trivial question $Q_{000\cdots000}$.} $Q_{\mu_1\cdots\mu_N}$, $\mu=0,1,2,3$, are pairwise independent and logically closed under $\leftrightarrow$ and thus form an informationally complete set $\cq_{M_N}$ with $D_N=4^N-1$. Moreover, $Q_{\mu_1\cdots\mu_N}$ and $Q_{\nu_1\cdots\nu_N}$ are
compatible if they differ by an {\bf even} number (incl.\ $0$) of non-zero indices and complementary otherwise.
\end{theorem}

We note that an $N$-qubit density matrix has precisely $4^N-1$ degrees of freedom.

\subsection{Linear, reversible time evolution and a quadratic information measure}\label{sec_lrtime}

Thus far, the summarized results invoked only rules \ref{lim} and \ref{unlim} (and in one instance rule \ref{loc}). Rules \ref{pres} and \ref{time}, on the other hand, can be demonstrated to entail a {\it linear} and {\it reversible} evolution of the generalized Bloch vector $\mathbb{R}^{4^N-1}\ni\vec{r}=2\,\vec{y}-\vec{1}$ that already appeared in (\ref{ansatz}),
\ba
\vec{r}(\Delta t+t_0)=T(\Delta t)\,\vec{r}(t_0),\label{linear}
\ea
where $T(\Delta t)\subset\ct_N$ defines a one-parameter matrix group \cite{Hoehn:2014uua}. Suppose $T(\Delta t),T'(\Delta t')\in\ct_N$ correspond to two distinct interactions to which $O$ may subject $S_N$. By rule \ref{time}, $T(\Delta t)\cdot T'(\Delta t')$ must likewise be contained in $\ct_N$ and since both $T,T'$ are invertible, also the entire set $\ct_N$ must be a group. We shall henceforth often represent states with Bloch vectors $\vec{r}$.

Rules \ref{pres} and \ref{time}, together with elementary operational conditions on the information measure, enforce it to be quadratic $\alpha(y_i)=(2\,y_i-1)^2$ so that $O$'s total information (\ref{infmeas})
\ba
I_N(\vec{y})=\sum_{i=1}^{4^N-1}(2\,y_i-1)^2=|\vec{r}|^2\label{infomeasure}
\ea 
 is simply the square norm of the Bloch vector \cite{Hoehn:2014uua}. Interestingly, this derivation would not work without the {\it continuity} of time evolution (rule \ref{time}). Crucially, (\ref{infomeasure}) is {\it not} the Shannon entropy (see \cite{Hoehn:2014uua} for a discussion about why the Shannon entropy is also conceptually not suitable for quantifying $O$'s information). This reconstruction thereby corroborates an earlier proposal for a quadratic information measure for quantum theory by Brukner and Zeilinger \cite{Brukner:2002kx,Brukner:vn,Brukner:ys}.
 
This quadratic information measure becomes key for the remaining steps of the reconstruction. Given that (\ref{infomeasure}) is a `conserved charge' of time evolution (rule \ref{pres}), we can already infer that $\ct_N\subset\rm{SO}(4^N-1)$ because time evolution must be connected to the identity.

\subsection{Pure and mixed states}

Suppose $O$ knows $S_N$'s answers to $N$ mutually compatible questions from $\cq_{M_N}$, thereby saturating the information limit of $N$ {\it independent} \texttt{bits} (rule \ref{lim}). He will then also know the answers to each of their bipartite, tripartite, ..., and $N$-partite XNOR conjunctions which, by theorem \ref{thm_nqbit}, are also in $\cq_{M_N}$ (and compatible). In total, he then knows the answers to 
\ba
\binom{N}{1}+\binom{N}{2}+\cdots\binom{N}{N}=\sum_{i=1}^{N}\,\binom{N}{i}=2^N-1\nn
\ea
questions from $\cq_{M_N}$. Thus, $O$'s total information (\ref{infomeasure}) is $2^N-1$ \texttt{bits} in this case. It contains {\it dependent} \texttt{bits} of information because the questions in $\cq_{M_N}$ are pairwise, but not all mutually independent. Thanks to rule \ref{pres}, this is invariant under time evolution.

This allows us to distinguish two kinds of states \cite{Hoehn:2014uua}; $\vec{y}$ is called a
\begin{description}
\item {\bf pure state:} if it is a state of maximal information, and hence of maximal length
\ba
I_{N}(\vec{y})=\sum_{i=1}^{4^N-1}\,(2\,y_i-1)^2=(2^N-1)\, \texttt{bits},\label{pure}
\ea
\item {\bf mixed state:} if it is a state of non-maximal information, 
\ba
0\,\texttt{bit}\leq I_{N}(\vec{y})=\sum_{i=1}^{4^N-1}\,(2\,y_i-1)^2<(2^N-1)\, \texttt{bits}. \label{mixed}
\ea
\end{description}
The square length of the Bloch vector thus corresponds to the number of answered questions. The state of no information $\vec{y}=\f{1}{2}\,\vec{1}$ has length $0$ \texttt{bits}.

As can be easily checked, quantum theory satisfies this characterization. In particular, an $N$-qubit density matrix, corresponding to a pure state, has a Bloch vector with square norm equal to $2^N-1$. This peculiar fact now has a clear informational interpretation.

\subsection{The Bloch ball and unitary group for a qubit from a conserved informational charge}

Since $D_1=3$ (cf.\ sec.~\ref{sec_3d}), we have that $\cq_{M_1}=\{Q_1,Q_2,Q_3\}$ is a maximal set of mutually complementary questions, i.e., no further $Q\in\cq_1$ can be added to $\cq_{M_1}$ without destroying mutual complementarity in the set (cf.\ sec.~\ref{sec_connect}). According to (\ref{pure}), a pure state satisfies
\ba
I_{N=1}(\vec{y})=r_1^2+r_2^2+r_3^2=(2\,y_1-1)^2+(2\,y_2-1)^2+(2\,y_3-1)^2=1\,\texttt{bit}\label{n1charge}.
\ea
For later, we thus observe: {\it for pure states, the maximal mutually complementary set carries \emph{exactly} $1$ \texttt{bit} of information and this is a conserved charge of time evolution (rule \ref{pres}).}
 
Rule \ref{lim} implies that, e.g., the pure state $\vec{y}_*=(1,0,0)$ exists in $\Sigma_1$ and we know $\ct_1\subset\rm{SO}(3)$. But it is clear that applying {\it any} $T\in\rm{SO}(3)$ to $\vec{y}_*$, according to (\ref{linear}), yields only states that are also compatible with all rules \ref{lim}--\ref{pres} (and the landscape). Hence, by rule \ref{time} we must actually have $\ct_1=\rm{SO}(3)\simeq\rm{PSU}(2)$. Clearly, $\ct_1$ then generates {\it all} quantum pure states from $\vec{y}_*$, i.e., it yields the entire Bloch sphere (the image of any legal state under a legal time evolution is also a legal state). Recalling that $\Sigma_1$ is convex, we obtain that $\Sigma_1=B^3\simeq\text{convex hull of }\mathbb{CP}^1$ is the entire unit Bloch ball with mixed states (\ref{mixed}) lying inside; the completely mixed state equals the state of no information at the center. $\Sigma_1,\ct_1$ coincide exactly with the set of density matrices $\rho=\f{1}{2}(\mathds{1}+\vec{r}\cdot\vec{\sigma})$ and the set of unitary transformations $\rho\mapsto U\,\rho\,U^\dag$, $U\in\rm{SU}(2)$, respectively, for a single qubit in its {\it adjoint} (i.e., Bloch vector) representation, where $\vec{\sigma}=(\sigma_1,\sigma_2,\sigma_3)$ is the vector of Pauli matrices. Finally, from the assumptions in sec.~\ref{sec_prob} and rule \ref{Q} it is also clear that $\cq_1=\{\vec{q}\in\mathbb{R}^3\,|\,|\vec{q}|^2=1\,\texttt{bit}\}\simeq\mathbb{CP}^1$. This coincides with the set of projectors $P_{\vec{q}}=\f{1}{2}(\mathds{1}+\vec{q}\cdot\vec{\sigma})$ onto the $+1$ eigenspaces of the Pauli operators $\vec{q}\cdot\vec{\sigma}$. Noting that 
\ba
\Tr(\rho\,P_{\vec{q}})=\f{1}{2}(1+\vec{r}\cdot\vec{q})\equiv Y(Q|\vec{y})\label{bornn1}
\ea
we also recover that (\ref{ansatz}) yields the Born rule for projective measurements. We thus have the claim of sec.~\ref{sec_axioms} for $N=1$ (for details see \cite{Hoehn:2014uua,hw}).

\subsection{Unitary group and density matrices for two qubits from conserved informational charges}

Also for $N=2$ it is rewarding to consider maximal mutually complementary sets within $\cq_{M_2}$. Using lemma \ref{lem1}, one can check that there are exactly {\it six} maximal complementarity sets containing five questions and {\it twenty} containing three \cite{hw}; e.g., two graphical representatives are:
\ba
\psfrag{x1}{\tiny $Q_{1}$}
\psfrag{y1}{\tiny $Q_{2}$}
\psfrag{z1}{\tiny\hspace{-.13cm} $Q_{3}$}
\psfrag{x2}{\tiny\hspace{-.1cm} $Q'_{1}$}
\psfrag{y2}{\tiny $Q'_{2}$}
\psfrag{z2}{\tiny $Q'_{3}$}
\psfrag{xx}{\tiny $Q_{11}$}
\psfrag{xy}{\tiny $Q_{12}$}
\psfrag{xz}{\tiny\hspace{-.23cm} $Q_{13}$}
\psfrag{yx}{\tiny $Q_{21}$}
\psfrag{yy}{\tiny $Q_{22}$}
\psfrag{yz}{\tiny $Q_{23}$}
\psfrag{zx}{\tiny $Q_{31}$}
\psfrag{zy}{\tiny $Q_{32}$}
\psfrag{zz}{\tiny $Q_{33}$}
\includegraphics[scale=0.25]{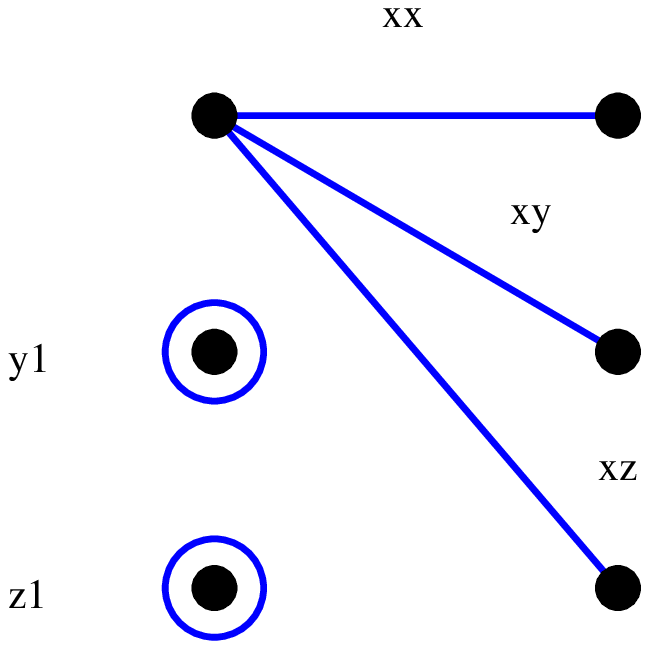} \q\q\q\q\q\q\q\q\q\q\q\q\q\q\q\q\includegraphics[scale=0.25]{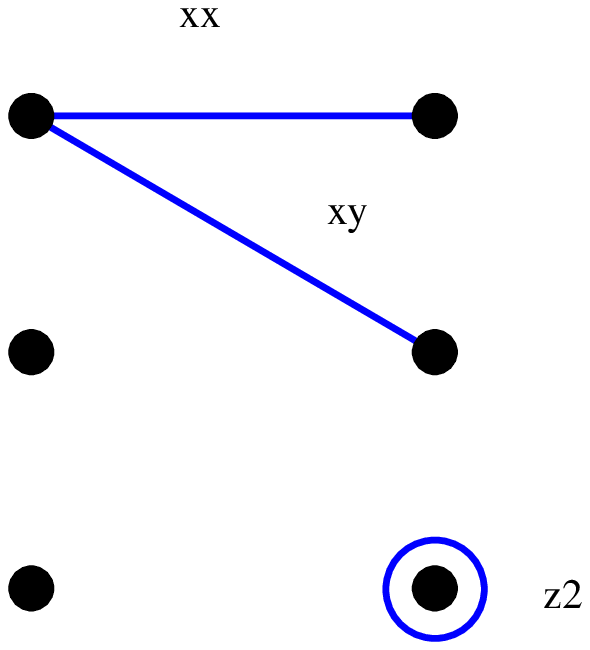} \q\q\q\nn\\\nn\\
  \text{Pent}_1=\{Q_{11},Q_{12},Q_{13},Q_{2},Q_{3}\}, \q\q\q\q\q\q\q\text{Tri}_1=\{Q_{11},Q_{12},Q'_3\}\,. \nn
\ea
The six maximal complementarity sets of five elements can be represented as a lattice of pentagons, see fig.~\ref{fig:pentagon} (which also contains four green triangles, each representing one of the twenty maximal complementarity sets of three questions) \cite{hw}.
\begin{SCfigure}%[h!]
%\begin{center}
\psfrag{x1}{\tiny $Q_{1}$}
\psfrag{y1}{\tiny $Q_{2}$}
\psfrag{z1}{\tiny\hspace{-.13cm} $Q_{3}$}
\psfrag{x2}{\tiny\hspace{-.1cm} $Q'_{1}$}
\psfrag{y2}{\tiny $Q'_{2}$}
\psfrag{z2}{\tiny $Q'_{3}$}
\psfrag{xx}{\tiny $Q_{11}$}
\psfrag{xy}{\tiny $Q_{12}$}
\psfrag{xz}{\tiny\hspace{-.23cm} $Q_{13}$}
\psfrag{yx}{\tiny $Q_{21}$}
\psfrag{yy}{\tiny $Q_{22}$}
\psfrag{yz}{\tiny $Q_{23}$}
\psfrag{zx}{\tiny $Q_{31}$}
\psfrag{zy}{\tiny $Q_{32}$}
\psfrag{zz}{\tiny $Q_{33}$}
\psfrag{p1}{\tiny $1$}
\psfrag{p2}{\tiny $2$}
\psfrag{p3}{\tiny $3$}
\psfrag{p4}{\tiny $4$}
\psfrag{p5}{\tiny $5$}
\psfrag{p6}{\tiny $6$}
\hspace*{-1cm}\includegraphics[scale=0.22]{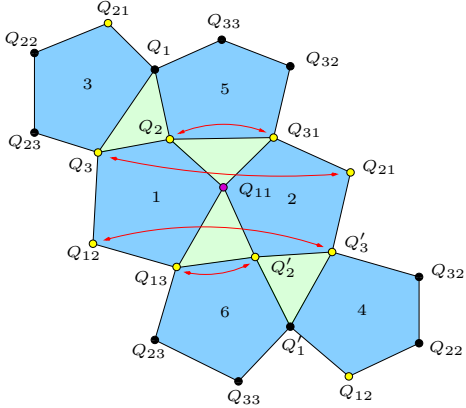}
\hspace*{1cm}\caption{\footnotesize The six maximal complementarity sets represented as pentagons. Two questions are complementary if they share a pentagon or are connected by an edge and compatible otherwise. Every pentagon is connected to all other five because any $Q\in\cq_{M_2}$ is contained in precisely two pentagons. The red arrows represent the information swap (\ref{eq:pentswapq}) between pentagons 1 and 2 that preserves all pentagon equalities (\ref{pentin}) and defines the time evolution generator (\ref{eq:pentswapT}). (Figure adapted from \cite{hw}.)
} \label{fig:pentagon}
 %\end{center}
 \end{SCfigure}

Each of these sets has to satisfy the complementarity inequalities (\ref{compstrong}); specifically $0\,\text{\texttt{bits}}\leq I(\text{Pent}_a):=\sum_{i\in\text{Pent}_a}r_i^2\leq 1\,\text{\texttt{bit}}$ for the information carried by the five questions in pentagon $a$. Since any $Q\in\cq_{M_2}$ is contained in precisely two pentagons (cf.\ fig.~\ref{fig:pentagon}) we find
\ba
\sum_{a=1}^6\,I(\text{Pent}_a)=2\left(\sum_{i=1,2,3}(r_{i_1}^2+r_{i_2}^2)+\sum_{i,j=1,2,3}r_{ij}^2\right)=2\,I_{{N=2}}(\vec{r}).\label{totinfo}
\ea
Noting that for pure states $I_{{N=2}}(\vec{r}_{\rm pure})=3$ \texttt{bits} thus produces the {\it pentagon equalities} \cite{hw}
\ba
\text{\bf pure states:}\q\q\q\q I(\text{Pent}_a)\equiv1\,\text{\texttt{bit}},\q\q a=1,\ldots,6.\label{pentin}
\ea
Any pure state {\it must} satisfy (\ref{pentin}) and $\ct_2$ evolves pure states to pure states (rule \ref{pres}). Hence, in analogy to $N=1$: {\it for pure states, these six maximal mutually complementary sets carry \emph{exactly} $1$ \texttt{bit} of information and these are six conserved charges of time evolution}. There are further interesting constraints on the distribution of $O$'s information over $\cq_{M_2}$ \cite{hw}.

It can be straightforwardly checked that quantum theory actually satisfies (\ref{pentin}). Indeed, in the case of quantum theory the identity for $\text{Pent}_1$ reads in more familiar language (pure states)
\ba
 I(\text{Pent}_1) = \langle \sigma_2\otimes\mathds{1}\rangle^2+\langle \sigma_3\otimes\mathds{1}\rangle^2+\langle \sigma_1\otimes\sigma_1\rangle^2+\langle\sigma_1\otimes\sigma_2\rangle^2+\langle\sigma_1\otimes\sigma_3\rangle^2=1,\nn
\ea
etc. Remarkably, these identities of quantum theory seem not to have been reported before in the literature. These novel conserved informational charges are a prediction of our reconstruction, underscoring the benefits of taking this informational approach. And these informational charges are indispensable for deriving the unitary group and the state space as we shall now see.

Using that $I(\text{Pent}_a(\vec{r}))$ is conserved under $\ct_2\subset\rm{SO}(15)$ entails (with new index $i=1,\ldots,15$)
\ba
 \sum_{i\in\text{Pent}_a,1\leq j\leq 15}r_i\,G_{ij}\,r_j=0,\q\q\q\q a=1,\ldots6, \label{eq:pentGeq}
\ea
where $T(\Delta t)=\exp(\Delta t G)$ for $G\in\fs\fo(15)$ \cite{hw}. The correlation structure  of fig.~\ref{fig_corr} enforces \cite{hw}
\ba
G_{ij}=0,\q\q\q\text{whenever $Q_i,Q_j$ are compatible.}\label{gij}
\ea
Each of the $15$ $Q_i\in\cq_{M_2}$ is complementary to eight others and since $G_{ij}=-G_{ji}$, there could be maximally $60$ linearly independent $G_{ij}$ of $\ct_2$. 

These are constructed as follows. For {\it every} pair of pentagons there is a unique information swap transformation which preserves (\ref{pentin}). For instance, the red arrows in fig.~\ref{fig:pentagon} represent the complete information swap between pentagons $\text{Pent}_1$ and $\text{Pent}_2$ ($\longleftrightarrow$ is {\it not} the XNOR)
\begin{equation}
r^2_{2}\longleftrightarrow r_{31}^2 \ (\text{Pent}_5), \ r_3^2\longleftrightarrow r_{21}^2 \ (\text{Pent}_3), \ r_{12}^2\longleftrightarrow {r'_3}^2 \ (\text{Pent}_4), \ r_{13}^2\longleftrightarrow r_2'^2 \ (\text{Pent}_6) \label{eq:pentswapq}
\end{equation}
that keeps all other components fixed. (\ref{pentin}) are preserved because every swap in (\ref{eq:pentswapq}) occurs within a pentagon. The correlation structure of fig.~\ref{fig_corr} fixes the corresponding generator to \cite{hw}
\begin{gather}
G_{ij}^{\text{Pent}_1,\text{Pent}_2}=\delta_{i2}\delta_{j(31)}-\delta_{i3}\delta_{j(21)}+\delta_{i(12)}\delta_{j3'}-\delta_{i(13)}\delta_{j2'}-(i\longleftrightarrow j). \label{eq:pentswapT}
\end{gather}
One can repeat the argument for all $15$ pentagon pairs, producing $15$ linearly independent generators \cite{hw}. Remarkably, they turn out to coincide exactly with the adjoint representation of the $15$ fundamental generators of $\rm{SU}(4)$ \cite{hw}. In particular, (\ref{eq:pentswapT}) is the generator of entangling unitaries leaving $r_{11}$ invariant.
The other $45$ independent generators satisfying (\ref{gij}) are ruled out by the correlation structure so that $\ct_2$ cannot be generated by anything else than these $15$ pentagon swaps \cite{hw}. One can show that the exponentiation of (linear combinations of) these $15$ pentagon swaps generates $\rm{PSU}(4)$ and that this group abides by all rules and forms a maximal subgroup of $\rm{SO}(15)$ \cite{hw}. Rule \ref{time} then implies $\ct_2\simeq\rm{PSU}(4)$ which is the correct set of unitary transformations $\rho\mapsto U\,\rho\,U^\dag$, $U\in\rm{SU}(4)$, for two qubits.

It turns out that the set of Bloch vectors satisfying all six pentagon equalities (\ref{pentin}) and the conservation equations (\ref{eq:pentGeq}) for the $15$ pentagon swaps splits into two sets on each of which $\ct_2=\rm{PSU}(4)$ acts transitively \cite{hw}. These two sets correspond precisely to the two possible conventions of building up composite questions either using the XNOR or XOR (cf.\ sec.~\ref{sec_connect}) and are therefore physically equivalent. Adhering to the XNOR convention, we conclude that the surviving set of Bloch vectors solving (\ref{pentin}, \ref{eq:pentGeq}) is the set of $N=2$ states admitted by the rules. Indeed, it coincides exactly with the set of quantum pure states which forms a $\mathbb{CP}^3$  of which $\rm{PSU}(4)$ is the isometry group \cite{hw}. Employing convexity of $\Sigma_2$, one finally finds
\ba
\Sigma_2= \text{closed convex hull of } \mathbb{C}\mathbb{P}^3,\nn
\ea
which is exactly the set of normalized $4\times4$ density matrices over $\mathbb{C}^2\otimes\mathbb{C}^2$. 

Concluding, the new conserved informational charges (\ref{pentin}), in analogy to (\ref{n1charge}) for $N=1$, define both the unitary group and set of states for two qubits. (For neglected details, see \cite{hw}.)

\subsection{Unitaries and states for $N>2$ elementary systems}

According to theorem \ref{thm_nqbit}, $\Sigma_N$ is $(4^N-1)$-dimensional and $\ct_N\subset\rm{SO}(4^N-1)$ (cf.\ sec.~\ref{sec_lrtime}). The reconstruction of the unitary group uses a {\it universality} result from quantum computation: two-qubit unitaries $\rm{PSU}(4)$ (between any pair) and single-qubit unitaries $\rm{PSU}(2)\simeq\rm{SO}(3)$ generate the full projective unitary group $\rm{PSU}(2^N)$ for $N$ qubits \cite{bremner2002practical,Harrow:2008aa}. Given that $S_N$ is a composite system, all of these bipartite and local unitaries must be in $\ct_N$. One can check that $\rm{PSU}(2^N)$ again abides by all rules and constitutes a maximal subgroup of $\rm{SO}(4^N-1)$ \cite{hw}. Thanks to rule \ref{time}, this yields $\ct_N\simeq\rm{PSU}(2^N)$ which coincides with the set of unitary transformations on $N$-qubit density matrices. In analogy to the previous case, one obtains as the state space
\ba
\Sigma_N= \text{closed convex hull of } \mathbb{C}\mathbb{P}^{2^N-1},\nn
\ea
which agrees with the set of normalized $N$-qubit density matrices. (For details, see \cite{hw}.)

\subsection{Questions as projective measurements and the Born rule}

The assumptions in sec.~\ref{sec_prob} and rule \ref{Q} yield the following question set characterization \cite{hw}:
\ba
\cq_N\simeq\{\vec{q}\in\mathbb{R}^{4^N-1}\,|\, Y(\vec{q}|\vec{r})\in[0,1]\,\forall\,\vec{r}\in\Sigma_N\q\text{and}\q\vec{q}\q\text{is a $1$ \texttt{bit} quantum state}\}.\label{Qcharacter}
\ea
As shown in \cite{hw}, this set is ismorphic to the set of projectors $P_{\vec{q}}=\f{1}{2}(\mathds{1}+\vec{q}\cdot\vec{\sigma})$ onto the $+1$ eigenspaces of the Pauli operators $\vec{q}\cdot\vec{\sigma}=\sum_{\mu_1\cdots\mu_N}\,q_{\mu_1\cdots\mu_N}\,\sigma_{\mu_1\cdots\mu_N}$, where $\sigma_{\mu_1\cdots\mu_N}=\sigma_{\mu_1}\otimes\cdots\otimes\sigma_{\mu_N}$ and $\sigma_0=\mathds{1}$. Noting that $q_{\mu_1\cdots\mu_N}$ corresponds to (\ref{ngbit}) reveals that the XNOR at the question level corresponds to the tensor product $\otimes$ at the operator level. One also finds that (\ref{bornn1}) again holds such that (\ref{ansatz}) yields the Born rule for projective measurements for arbitrary $N$. (For the neglected details and many further interesting properties of $\cq_{N}$, we refer to \cite{hw}.)

\subsection{The von Neumann evolution equation}

We thus obtain qubit quantum theory in its adjoint (i.e.\ Bloch vector) representation. Lastly, we note that $\vec{r}(t)=T(t)\,\vec{r}(0)$ with $T(t)=e^{t\,G}\in\rm{PSU}(2^N)$ is equivalent to the adjoint action 
\ba
\rho(t)=U(t)\,\rho(0)\,U^\dag(t),\label{unitary}
\ea 
of $U(t)=e^{-i\,H\,t}\in\SU(2^N)$ for some hermitian operator $H$ on $\mathbb{C}^{2^N}$, where $\rho(t)=\f{1}{2^N}\left(\mathds{1}+\vec{r}(t)\cdot\vec{\sigma}\right)$ \cite{hw}. (\ref{unitary}), in turn, is equivalent to $\rho(t)$ solving the von Neumann evolution equation 
\ba
i\f{\p\,\rho}{\p t}=[H,\rho].
\ea
We have therefore also recovered the correct time evolution equation for quantum states.

\section{Conclusions}\label{sec_conc}

We have reviewed and summarised the key steps from \cite{Hoehn:2014uua,hw} necessary to prove the claim of sec.~\ref{sec_axioms}. This yields a reconstruction of the explicit formalism of qubit quantum theory from rules constraining an observer's acquisition of information about a system \cite{Hoehn:2014uua,hw}. The derivation corroborates the consistency of interpreting the state as the observer's `catalogue of knowledge' and shows that it is sufficient to speak only about the information accessible to him for reproducing quantum theory. In fact, for qubits, this derivation accomplishes an informational reconstruction of the type proposed in Rovelli's {\it relational quantum mechanics} \cite{Rovelli:1995fv} and in the Brukner-Zeilinger informational interpretation of quantum theory \cite{zeilinger1999foundational,Brukner:2002kx}.

%Aside f
%
%%corroborates informational interpreatation as consistent, best thing to get interpretation by recon and completes RQM...
%%
%%informational, i.e. no info about concrete physics, but especially engineered to expose architecture. novel insights.
%
%Our approach, in fact, provides a novel formulation of quantum theory which generates
%
%Our approach is especially engineered to expose quantum theory's architecture and provides

As a key benefit, this reconstruction also provides a novel informational explanation for the architecture of qubit quantum theory. In particular, it explains the logical structure of a basis of spin measurements, the dimensionality and structure of quantum state spaces, the correlation structure and the unitarity of time evolution from the perspective of information acquisition. This unravels previously unknown structural properties: conserved `informational charges' from complementarity relations define and explain the unitary group and the set of pure states.

%\section*{Acknowledgments}
\ack

The author thanks C.\ Wever for an enjoyable collaboration on \cite{hw}. The project leading to this publication has received funding from the European Union's Horizon 2020 research
and innovation programme under the Marie Sklodowska-Curie grant agreement No 657661.

\end{document}